\documentclass[ssy,preprint]{imsart}

\RequirePackage{amsthm,amsmath,amssymb,latexsym}
\RequirePackage{graphicx}
\RequirePackage[authoryear]{natbib}
\bibpunct{(}{)}{;}{a}{}{;}
\RequirePackage[colorlinks,citecolor=blue,urlcolor=blue,breaklinks]{hyperref}

\arxiv{arXiv:1610.09572}

\startlocaldefs
\numberwithin{equation}{section}
\theoremstyle{plain}
\newtheorem{theorem}{Theorem}[section]
\newtheorem{corollary}{Corollary}[section]
\newtheorem{definition}{Definition}[section]
\newtheorem{lemma}{Lemma}[section]
\endlocaldefs

\newcommand{\irm}{\ensuremath \mathrm{i}}
\newcommand{\sech}{\ensuremath \operatorname{sech}}
\newcommand{\diagdots}[3][-10]{%
  \rotatebox{#1}{\makebox[0pt]{\makebox[#2]{\xleaders\hbox{$\cdot$\hskip#3}\hfill\kern0pt}}}%
}
\newcommand{\diagdotss}[3][-14]{%
  \rotatebox{#1}{\makebox[0pt]{\makebox[#2]{\xleaders\hbox{$\cdot$\hskip#3}\hfill\kern0pt}}}%
}

\begin{document}

\begin{frontmatter}
\title{Density Tracking by Quadrature for Stochastic Differential Equations}
\runtitle{Density Tracking by Quadrature for SDE}
\begin{aug}
\author{\fnms{Harish S.} \snm{Bhat}\thanksref{m1,t2}\ead[label=e1]{hbhat@ucmerced.edu}} \and
\author{\fnms{R. W. M. A.} \snm{Madushani}\thanksref{m2}\ead[label=e2]{Anusha.WasalaMudiyanselage@medicine.ufl.edu}}
\thankstext{t2}{Partially supported by the National Science Foundation award DMS-1723272.}
\runauthor{Bhat and Madushani}

\affiliation{University of California, Merced\thanksmark{m1} and University of Florida\thanksmark{m2}}

\address{Harish S. Bhat\\
School of Natural Sciences\\
5200 N. Lake Rd.\\
Merced, CA 95343 USA\\
\printead{e1}}

\address{R. W. M. A. Madushani\\
Department of Medicine\\
University of Florida\\
Gainesville, FL 32610 USA\\
\printead{e2}}
\end{aug}

\begin{abstract}
We develop and analyze a method, density tracking by quadrature
(DTQ), to compute the probability density function of the solution of
a stochastic differential equation.  The
derivation of the method begins with the discretization in time of the
stochastic differential equation, resulting in a discrete-time Markov
chain with continuous state space. At each time step, DTQ
applies quadrature to solve the Chapman-Kolmogorov equation for this
Markov chain.  In this paper, we focus on a particular case of the DTQ
method that arises from applying the Euler-Maruyama method in time and
the trapezoidal quadrature rule in space.  Our main result establishes
that the density computed by DTQ converges in $L^1$ to both 
the exact density of the
Markov chain (with exponential convergence rate), and to the exact
density of the stochastic differential equation (with first-order 
convergence rate).   We establish a Chernoff bound that implies
convergence of a domain-truncated version of DTQ.   We
carry out numerical tests to show that the empirical performance of
DTQ matches theoretical results, and also to
demonstrate that DTQ can compute
densities several times faster than a Fokker-Planck solver,
for the same level of error.
\end{abstract}



\end{frontmatter}

\section{Introduction}
\label{sect:intro}
Suppose that $(\Omega, \mathcal{F}, \{\mathcal{F}_t\}_{t \geq 0}, \mathbb{P})$ is a complete probability space such that the filtration $\{\mathcal{F}_t\}_{t \geq 0}$ satisfies the usual conditions.  Let $W_t$ denote the Wiener process defined on the probability space.  Consider the scalar stochastic differential equation (SDE)
\begin{equation}
\label{sde}
dX_t = f(X_t)dt + g(X_t)dW_t.
\end{equation}
For simplicity, we assume a deterministic initial condition $X_0 = C$.  Note that $X_t$ is an It\^{o} diffusion; neither the drift $f$ nor the diffusion $g$ feature explicit time-dependence.  Assuming regularity of $f$ and $g$, $X_t$ has a probability density function $p(x,t)$ \citep{Rogers1985}.  In this paper, we develop a convergent numerical method to solve for $p$.  We call our method density tracking by quadrature (DTQ).

To introduce DTQ informally, let us describe the three main steps in its derivation:
\begin{enumerate}
\item Discretize the SDE (\ref{sde}) in time.
\item Interpret the time-discretized equation as a discrete-time
  Markov chain; let $\tilde{p}$ denote its density.  Write the
  Chapman-Kolmogorov equation for the time-evolution of $\tilde{p}$.
\item Discretize both the Chapman-Kolmogorov equation and $\tilde{p}$ in space, e.g., using a spatial grid and numerical quadrature.  Let $\hat{p}$ denote the discrete-space approximation of $\tilde{p}$.
\end{enumerate}
We use in step 1 the explicit Euler-Maruyama method and the trapezoidal rule in step 3; unless stated otherwise, this is the DTQ method analyzed in this paper.  Please note that the above steps give a blueprint for many possible algorithms; it is entirely possible that by choosing a different time integrator and a different quadrature rule, one could derive a DTQ method that improves upon the default method studied here.

In this paper, we prove that $\hat{p}$ converges to
$p$ as the discretization parameters tend to zero.  Because there are
existing results on the convergence of $\tilde{p}$ to $p$, the main
task of this paper is to show that $\hat{p} \to \tilde{p}$.

\citet{BallyTalay1996} established conditions under which $\tilde{p}$
converges to $p$, in the case where Euler-Maruyama is used
to discretize (\ref{sde}) in time.  Let $\| f \|_1$ denote the
$L^1$ norm of a function $f$.  Suppose we seek the density of
(\ref{sde}) at time $T > 0$.  Let $h > 0$ denote the temporal step
size; as we take $h \to 0$, we assume $T = N h$ stays fixed.  Then the
results of \citet{BallyTalay1996} imply that $\| p(\cdot,T) -
\tilde{p}(\cdot,T) \|_1 = O(h)$.

Our work builds on this result.  The DTQ method analyzed here combines
Euler-Maruyama temporal discretization with the trapezoidal rule on an
equispaced grid.  This results in a fast, simple method to compute an
approximation $\hat{p}$ such that $\| \tilde{p}(\cdot,T) -
\hat{p}(\cdot,T) \|_1 = O(h^{-1} \exp(-r h^{-\kappa}))$ for positive
constants $r$, $\kappa$.  The user can control $\kappa$
by adjusting the relationship between the spatial and temporal grid
spacings.

The primary application of this work that we envision is in
statistical estimation and inference for diffusion processes.  DTQ can be
used to numerically approximate the likelihood
function for a diffusion that is observed at discrete points in
time \citep{BhatMadu2016, BhatMaduRawat2016}.  The present work lays 
a theoretical foundation for these statistical applications.  
Additionally, note that when
estimation/inference procedures for diffusions have been compared, 
a method that
approximates the likelihood by numerically solving the Fokker-Planck
(or Kolmogorov) equation achieves superior accuracy at the cost of
excessive computational time \citep{Hurn2007}.  The results of the
present paper indicate that DTQ achieves the
same accuracy as a Fokker-Planck solver with less computational
effort, further motivating its use.

We now review alternative approaches to compute the density of (\ref{sde}), including prior work on DTQ and its relatives.

\subsection{Alternative Approaches}
\label{sect:alternativeapproaches}
If the drift $f$ and diffusion $g$ are sufficiently smooth, then $p$
satisfies the forward Kolmogorov (or Fokker-Planck) equation \citep{Rogers1985}:
\begin{equation}
\label{eqn:kolmo}
\frac{\partial}{\partial t}p(x,t) = -\frac{\partial}{\partial x}[f(x,t)p(x,t)]+\frac{1}{2}\frac{\partial^2}{\partial x^2}[g^2(x,t)p(x,t)].
\end{equation}
Prescribing an initial condition $p(x,0)$, we may then solve
(\ref{eqn:kolmo}) to obtain the density $p(x,T)$ at time $T > 0$.  The
solution of (\ref{eqn:kolmo}) must satisfy the normalization
condition $\int_{x \in \mathbb{R}} p(\cdot,t) \, dx = 1$, which
implies boundary conditions of the form $\lim_{|x| \to \infty} p(x,t)
= 0$.

We view DTQ as an alternative to numerical methods for the
solution of (\ref{eqn:kolmo}).  The primary purpose of the present paper is to
demonstrate intrinsic properties---both theoretical and
empirical---of DTQ.  We compare DTQ with a 
finite difference method for the solution of (\ref{eqn:kolmo}); this is a
logical choice given the particular version of DTQ studied here.
In the present version, the density is numerically
approximated (i.e., finite-dimensionalized) by a sequence of values on an
equispaced grid, just as in a finite difference method for a partial 
differential equation (PDE).  By instead choosing to represent the
unknown density as an expansion in a basis or frame, we can derive different
versions of the DTQ method that are akin to finite element, meshless, and 
Hermite spectral methods for (\ref{eqn:kolmo})
\citep{DiPaola2002, Pichler2013, Canor2013, Luo2013}.  We will pursue this
line of reasoning and resulting comparisons in future work.  In the present
work, we compare DTQ against a finite difference Fokker-Planck solver that is
first-order in time and second-order in space. For a particular test
problem at the finest grid resolution we consider, DTQ
computes a solution with $L^1$ error $\approx 3 \times 10^{-3}$ more
than $100$ times faster than the Fokker-Planck solver.

% In the univariate case, this method approximates the transition density $p$ using a Hermite function expansion. Prior work has established convergence and shown that the method works well for models of financial interest, requiring only a small number of terms \citep{AitSahalia2002}.  In the present work, we compare likelihood surfaces computed via DTQ and the closed-form Hermite approximation method---see Section \ref{sect:summaryandoutline} for a summary of these and other results.

Both numerical Fokker-Planck solvers and DTQ are deterministic approaches that avoid random sampling.  We also place in this category the closed-form approximation methods of A\"it-Sahalia for both univariate and multivariate \citep{AitSahalia2002, AitSahalia2008} diffusions.  As opposed to deterministic approaches, one might try to estimate the density of (\ref{sde}) by sampling.  Specifically, one can employ any convergent numerical method to step (\ref{sde}) forward in time from $t=0$ to $t=T$, thereby generating one sample of $X_T$.  Repeating this procedure many times, one can obtain enough samples of $X_T$ to compute a statistical estimate of the density at time $T$.  For instance, one could compute a histogram or a kernel density estimate.  Several existing methods can be viewed as special cases and/or extensions of this approach \citep{HuWatanabe1996, KohatsuHiga1997, Milstein2004, Giles2015}.  In such methods, the accuracy of the density will be controlled by two parameters: the temporal step size and the number of sample paths.  If there are $N_S$ samples, then a typical stochastic time-stepping method will contribute an error of $N_S^{-1/2}$ and kernel density estimation will contribute an error of, e.g., $N_S^{-4/5}$.  In comparison, the DTQ method's accuracy is also controlled by two parameters, the temporal step size and the grid spacing.  Note that the trapezoidal rule on the real line contributes an error that decays exponentially in the grid spacing \citep{Trefethen2014}.  For this reason, we believe DTQ will be a strong alternative to a sampling-based method.

Returning to the forward Kolmogorov or Fokker-Planck equation (\ref{eqn:kolmo}), we see that smoothness of $f$ and $g$ is required in order to have classical solutions.  The implementation of DTQ itself does not utilize derivatives (whether exact or approximate) of $f$ and $g$.  At the same time, our convergence theory assumes analyticity of $f$ and $g$ on a strip in the complex plane that contains the real line.  We give two reasons for assuming analyticity.  First, many models of scientific interest involve functions $f$ and $g$ that do satisfy these hypotheses.  Second, in order to apply exponential error estimates for the trapezoidal rule \citep{Trefethen2014}, it is essential that our integrand, which depends on $f$ and $g$, be analytic on a strip.  Ultimately, we expect that the hypotheses in the present convergence proof can be relaxed, both by changing the quadrature scheme/estimates and by making use of improved estimates for the convergence of $\tilde{p}$ to $p$ \citep{GobetLabart2008}.  Still, the present results are sufficient for statistical tasks we have in mind.  

%Let $\mathring{p}$ be an approximate density that is computed in exactly the same way as $\hat{p}$ except for truncation of the infinite domain/series.  Our empirical results clearly show first-order convergence of $\mathring{p}$ to $p$, even when not all of the hypotheses of our theorem are satisfied.  Suppose that, inspired by these results, we discover how to prove convergence of $\hat{p}$ to $\tilde{p}$ assuming, for instance, that both $f$ and $g$ possess merely $4$ bounded continuous derivatives.  This will not immediately improve our ability to conclude that $\hat{p}$ converges to $p$; the existing result on convergence of $\tilde{p}$ to $p$ requires that both $f$ and $g$ are $C^\infty$ with bounded derivatives of all orders \citep{BallyTalay1996}.  To make true progress on the problem, we must relax the conditions of convergence for both $\hat{p} \to \tilde{p}$ and $\tilde{p} \to p$.  This is outside the scope of the present work.

\subsection{Prior Work}
\label{sect:priorwork}
DTQ has been described previously as numerical path integration.  The method has achieved accurate results on a variety of problems in, e.g., nonlinear mechanics and finance---see \citet{Wehner1983, Naess1993, Linetsky1997, Yu1997, RosaClot2002, Skaug2007}.  Recently, numerical path integration has been studied using semigroup methods \citep{Chen2017}; though convergence of $\tilde{p}$ to $p$ in $L^1$ is established, a fully discrete scheme (i.e., discretized in both time and space) is not analyzed.  Interestingly, \citet{Chen2017} do not require that the drift $f$ or diffusion $g$ are bounded above, nor do they require more than $4$ continuous derivatives for either function.  These results complement ours, especially as we seek in future work to relax hypotheses and to improve our error estimates for quantities computed in practice, i.e., $\hat{p}$ and its truncated domain version, $\mathring{p}$.

The DTQ method proposed here is an outgrowth of prior work on computing densities for stochastic delay differential equations \citep{BhatKumar2012, Bhat2014, BhatMadushani2015}.  The method from \citet{BhatMadushani2015}, when adapted to equations with no time delay, is the method in the present paper.  Our prior works did not address convergence from a theoretical standpoint, nor did they present empirical results of monotonic convergence that are in strict accordance with theory.  The present paper addresses both of these issues.

When we derive the DTQ method, we make use of the fact that a time-discretization of (\ref{sde}) can be viewed as a discrete-time Markov chain on a continuous state space.  Suppose we were to take a different point of view, that of trying to design a discrete-time Markov chain on a discrete state space whose law or density approximates well that of the original SDE.  In this case, there are extensive results starting from the work of \citet{Kushner1974}.  Like a discrete-time, discrete-time Markov chain, the DTQ algorithm can be written in the form $\hat{p}(t_{n+1}) = A \hat{p}(t_n)$, where $A$ is a matrix (possibly with an infinite number of rows and columns) and $\hat{p}(t_j)$ represents the approximate density at time $t_j$.  However, because of the quadrature-based derivation of the DTQ algorithm, the matrix $A$ is, in general, not a Markov transition matrix.  We find it both mathematically interesting and practically useful that, in spite of this, the DTQ method's $\hat{p}$ converges exponentially to $\tilde{p}$.  

The Chapman-Kolmogorov equation that is at the center of this paper---see (\ref{cdt})---has appeared in \citet{Pedersen1995,SantaClara1997}.  In these works, the right-hand side of the Chapman-Kolmogorov equation is interpreted as an expected value that can be computed using Monte Carlo methods.  In our approach, we use deterministic quadrature to evaluate the right-hand side of the Chapman-Kolmogorov equation.  There is one prior paper we found that features this approach, albeit in a different context, that of a nonlinear autoregressive time series model \citep{Cai2003}. The convergence results in \citet{Cai2003} are of a different nature than ours, because they involve taking the continuum limit in space but \emph{not} in time.  In the present work, we are interested in the error made by the DTQ method as both the temporal and spatial grid spacings vanish.  

\subsection{Summary of Results and Outline}
\label{sect:summaryandoutline}
The main result of this paper is a provably convergent method for
computing an approximation $\hat{p}$ of the density $p$ for the SDE
(\ref{sde}).  Let $h > 0$ and $k > 0$ denote, respectively, the temporal
and spatial step sizes.   Assume that $k \propto h^\rho$ for $\rho >
1/2$, and assume that $f$ and $g$ are sufficiently regular (more
precisely, admissible in the sense of Definition
\ref{def:admissible}).  Under these conditions, in Sections
\ref{sect:prelim} and \ref{sect:convthm}, we prove that $\hat{p}$
converges to $\tilde{p}$ in $L^1$, and that the error decays
exponentially in $h$.  Specifically, there exists a constant $r >
0$ such that the leading order $L^1$ error
term is proportional to $h^{-1} \exp(-r h^{1/2 - \rho})$---see
Theorem \ref{thm:convergence}.  As a consequence of this result and
the results of \citet{BallyTalay1996}, we conclude that $\hat{p}$ converges
to $p$ in $L^1$, and that the error decays linearly with $h$---see
Corollary \ref{cor:phatp}.

Up to and including Section \ref{sect:convthm}, our results pertain to
an idealized version of the DTQ algorithm in which we track the
density $\hat{p}$ at an infinite number of discrete grid points.  In
Section \ref{sect:trunc}, we study the effect of boundary
truncation.  Our main tool in this section is a Chernoff bound on the
tail sum of $\hat{p}$ that we establish through the moment generating function.
Let $\mathring{p}$ denote the approximation of $\hat{p}$
obtained by summing over precisely $2M+1$ grid points from $-y_M =
-Mk$ to $y_M = Mk$.  The quantity $\mathring{p}$ is what we actually
compute when we implement DTQ.  In Lemma \ref{lem:rL1}, we
show that if $y_M \to \infty$ at a logarithmic rate, i.e., $y_M
\propto \log h^{-1}$, then the $L^1$ error between $\mathring{p}$ and
$\hat{p}$ is $O(h)$.  Combining this with our earlier results, this
establishes $L^1$ convergence of $\mathring{p}$ to the true density
$p$---see Corollary \ref{cor:everything}.

In Section \ref{sect:numexp}, we study the performance of the DTQ
method.  For a suite of six test problems for which we have access to
the exact solution, our numerical tests confirm $O(h)$
convergence of $\mathring{p}$ to $p$.  This remains true for drift $f$
and diffusion functions $g$ that do not strictly satisfy the
hypotheses of our convergence theory.  We also present a finite
difference method for solving (\ref{eqn:kolmo}); we compare this method
against three different implementations of DTQ, and find that
DTQ is competitive.

% In Section \ref{sect:numexp}, we also explain how to use DTQ to compute likelihoods given an SDE model and data.  For two test problems, we compare the likelihood surfaces computed via DTQ against those computed using the closed-form Hermite approximation method \citep{AitSahalia2002}.  For an Ornstein-Uhlenbeck test problem, both methods' likelihood surfaces are similar and would yield reasonable parameter estimates.  For a test problem featuring bistability and bimodality, DTQ outperforms the closed-form method.  In this case, maximum likelihood estimates obtained via DTQ are far closer to the ground truth than those obtained from the closed-form method.  

Before proceeding, we give a more detailed derivation of the DTQ method in Section \ref{sect:setup} and then introduce necessary assumptions and notation in Section \ref{sect:notation}.

\section{Problem Setup}
\label{sect:setup}
We begin with a more detailed derivation of the DTQ method. First, we discretize (\ref{sde}) in time using the explicit Euler-Maruyama method:
\begin{equation}
\label{em}
x_{n+1} = x_n+f(x_{n}) h +g(x_{n})\sqrt{h}Z_{n+1},
\end{equation}
where $h > 0$ is a fixed time step and $Z_{n+1}$ is a random variable with a standard (mean zero, variance one) Gaussian distribution. We let $\tilde{p}(x,t_n)$ denote the probability density function of $x_n$.  Note that this differs from $p(x,t_n)$.  

From (\ref{em}), we observe that the density of $x_{n+1}$ given $x_n =
y$ is Gaussian with mean $y+f(y) h$ and variance $h g^2(y)$.  Let us
denote this conditional density by $\tilde{p}_{n+1, n}(x | y)$; then
\begin{equation}
\label{eqn:Gdef}
\tilde{p}_{n+1, n}(x | y)= G(x,y) := 
\frac{1}{\sqrt{2\pi g^2(y) h}} \exp\left( -\frac{(x-y-f(y) h)^2}{2g^2(y) h} \right).
\end{equation}
Note that, for any $y \in \mathbb{R}$,
\begin{equation}
\label{eqn:Gint}
\int_{x \in \mathbb{R}} G(x,y) \, dx = 1.
\end{equation}
With these definitions, we obtain
\begin{equation}
\label{cdt}
\tilde{p}(x, t_{n+1})=\int_{y \in \mathbb{R}} \tilde{p}_{n+1, n}(x | y ) \tilde{p}(y, t_n) \, dy,
\end{equation}
\noindent
the Chapman-Kolmogorov equation for the discrete-time, continuous-space Markov chain given by (\ref{em}).  Similar equations are often employed in the literature on inference for diffusions---see \citet{Pedersen1995}; \citet{SantaClara1997}; \citet[Chap 6.3.3]{fuchs2013inference}; and \citet{Kou2012}.

Let us define an equispaced temporal grid by $t_n = n h$ with $h =
T/N$.  In principle, we can now repeatedly apply (\ref{cdt}) to
determine $\tilde{p}(x,T)$.  This assumes we can perform the integral
over the real line.

To compute (\ref{cdt}), we use numerical quadrature.  Here
we employ the trapezoidal rule, enabling the use of exponential
error estimates \citep{Trefethen2014,Stenger1993,LundBowers1992}.  To
begin with, we apply the trapezoidal rule on the real line.  Later, we explain how to incorporate the effects of a finite, truncated integration domain.

Assume the domain $\mathbb{R}$ is discretized via an equispaced grid $y_j = j k$ where $k > 0$ is fixed.  Then our discrete-time, discrete-space evolution equation is
\begin{equation}
\label{eqn:phatn1}
\hat{p}(x,t_{n+1}) = k \sum_{j=-\infty}^\infty G(x,y_j) \hat{p}(y_j,t_n).
\end{equation}
Except for the fact that we have not yet truncated the infinite sum, this is the DTQ method.

%Thus far we have avoided the discussion of initial conditions for both
%$\tilde{p}$ and $\hat{p}$.  
In what follows, we assume a constant initial condition $X_0 = C$, which 
implies
$p(x,0)=\tilde{p}(x,0)=\delta(x-C)$.  This choice is not essential to
either the use or convergence of the DTQ method.  In fact, the choice
of a point mass initial condition requires special handling, because
we cannot discretize $\tilde{p}(x,0)$ directly.  We insert $n=0$
into (\ref{cdt}), use $\tilde{p}(x,0) = \delta(x-C)$, and obtain
the non-singular initial condition
\begin{equation}
\label{eqn:nonsingic}
\hat{p}(x,t_1) = \tilde{p}(x,t_1) = G(x,C).
\end{equation}
This enables us to iteratively use (\ref{eqn:phatn1}) for $n \geq 1$.

Our main task in Sections \ref{sect:prelim} and \ref{sect:convthm}
is to estimate $\| \hat{p}(\cdot,T) - \tilde{p}(\cdot,T) \|_1$.  Before we start the proof of Theorem \ref{thm:convergence}, we introduce necessary notation and assumptions.

\section{Notation and Assumptions}
\label{sect:notation}
We will use the Roman $\irm$ for the imaginary unit ($\irm =
\sqrt{-1}$) and reserve the Italic $i$ for an index of summation.  We denote the $L^1$ norm of a function $f : \mathbb{R} \to \mathbb{R}$ by
$$
\| f \|_1 = \int_{x \in \mathbb{R}} |f(x)| \, dx.
$$
We denote the $\ell^1$ norm of the sequence $\{\omega_j\}_{j=-\infty}^\infty$ by
$$
\| \omega \|_{\ell^1} = \sum_{j=-\infty}^\infty |\omega_j|.
$$
For a function $f : \mathbb{R} \to \mathbb{R}$, we understand $\| f
\|_{\ell^1}$ to be the norm of the sequence obtained by applying $f$
on a spatial grid:
$$
\| f \|_{\ell^1} = \sum_{j=-\infty}^\infty | f(j k) |,
$$
where again $k > 0$ denotes the grid spacing.  We use 
$\lceil x \rceil$ to denote the smallest integer greater than or equal
to $x$, and $\lfloor x \rfloor$ to denote the largest integer less
than or equal to $x$.
The following definition is from the literature \citep{LundBowers1992}.
\begin{definition}
\label{def:bsd}
For $a > 0$, let $S_a$ denote the infinite strip of width $2a$ given by
$$
S_a = \left\lbrace z \in \mathbb{C}: \ |\Im(z)|<a \right\rbrace.
$$
Then $B(S_a)$ is the set of functions such that $\varphi \in B(S_a)$ iff
$\varphi$ is analytic in $S_a$,
\begin{equation}
\label{eqn:bsddecay}
\int_{-a}^{a} |\varphi(x + \irm y)| \, dy = O(|x|^\alpha), \ \ x \rightarrow \pm \infty, \ \ 0 \leq \alpha <1,
\end{equation}
and
\begin{equation}
\label{eqn:bsdint}
\mathcal{N}(\varphi, S_a) \equiv \lim _{y \rightarrow a^{-}} \biggl\lbrace  \int_{\mathbb{R}}  |\varphi(x+ \irm y)| \, dx  +  \int_{\mathbb{R}}  |\varphi(x- \irm y)| \, dx  \biggr\rbrace < \infty.
\end{equation}
\end{definition}
The next definition encapsulates the constraints that the coefficient functions $f$ and $g$ in the original SDE (\ref{sde}) must satisfy in order for us to show exponential convergence of $\hat{p}$ to $\tilde{p}$.
\begin{definition}
\label{def:admissible}
In this paper, we say that $f$ and $g$ are \emph{admissible} if they
satisfy the following properties.  First, there exists $d > 0$ such
that $f$ and $g$ are analytic on the strip $S_d$.  Additionally, there
exist positive, finite, real constants $M_1$, $M_2$, $M_3$, and $M_4$
such that for all $z \in S_d$,
\begin{subequations}
\label{eqn:assumptions}
\begin{gather}
\label{eqn:lipschitz}
|f'(z)| \leq M_1 \\
\label{eqn:gbound}
M_2 \leq |g(z)| \leq M_3 \\
\label{eqn:gnz}
\Re(g(z)) \neq 0 \\
\label{eqn:glip}
|g'(z)| \leq M_4.
\end{gather}
\end{subequations}
\end{definition}
We now state a theorem that gives an exponential error estimate for the trapezoidal rule \citep{LundBowers1992}, one that we shall use to bound the error made in one step of the DTQ method.  Other error estimates, with different hypotheses, can be found in the literature \citep{Stenger1993, Trefethen2014}.
\begin{theorem}
\label{thm:trap}
Suppose $\varphi \in B(S_a)$ and $k>0$.  Let
$$
\eta = \int_{\mathbb{R}} \varphi(x) \, dx
- k \sum_{j=-\infty}^{\infty} \varphi(j k).
$$
Then
$$
|\eta| \leq \frac{\mathcal{N}(\varphi,S_a)}{2 \sinh(\pi a/k)} \exp{(-\pi a/k)}.
$$
\end{theorem}
\begin{proof}See \citet[Theorem 2.20]{LundBowers1992}.  \end{proof}

\section{Preliminary Estimates}
\label{sect:prelim}
In this section, we prove several lemmas that are essential ingredients for the convergence theorem in Section \ref{sect:convthm}.  The overall goal of these lemmas is to show that the integrand
\begin{equation}
\label{eqn:ourintegrand}
\varphi(x,y,t_n) = G(x,y) \hat{p}(y,t_n),
\end{equation}
considered as a function of $y$ for the purposes of quadrature, satisfies the hypotheses of Theorem \ref{thm:trap}.

The first lemma enables us to pass from an estimate of the error made in one time step to an estimate of the error made across a non-zero interval of time, even as the number of time steps becomes infinite.
\begin{lemma}
\label{lem:o1t}
Suppose for the function $\xi : \mathbb{R}^+ \to \mathbb{R}^+$ there exist 
$\gamma > 1$, $\epsilon > 0$ and $h_0 > 0$ such that 
$\xi(h) \leq \epsilon h^\gamma$
for all $h < h_0$.  Fix $T > 0$ and define $h = T/N$ where $N \in \mathbb{N}^+$.  Then 
$$
\lim_{N \to \infty} \left[ h \sum_{j=1}^{N-1} (1 + \xi(h))^{j-1} \right] = T.
$$
\end{lemma} \begin{proof}Take $N$ sufficiently large so that $h < 1$ and $h < h_0$.  Then 
$$
\sum_{j=1}^{N-1} (1 + \xi(h))^{j-1} =  \xi(h)^{-1} \left[ (1 + \xi(h))^{N-1} - 1 \right] = \sum_{j=1}^{N-1} \binom{N - 1}{j} \xi(h)^{j-1}.
$$
Hence
$$
h \sum_{j=1}^{N-1} (1 + \xi(h))^{j-1} - T = -h + h \sum_{j=2}^{N-1} \binom{N - 1}{j} \xi(h)^{j-1},
$$
implying
$$
-h \leq h \sum_{j=1}^{N-1} (1 + \xi(h))^{j-1} - T \leq \sum_{j=2}^{N-1} \frac{T^j \epsilon^{j-1}}{j!} h^{(\gamma-1)(j-1)} \leq \epsilon^{-1} h^{\gamma - 1} \exp(T \epsilon).
$$
We have shown that the limit is $T$, and that the correction term to the limit is $O(h^{\gamma-1})$.  \end{proof}
The following lemma specializes an $\ell^1$-norm estimate of a discrete Gaussian to the case of our function $G$.
\begin{lemma}
\label{lem:discgauss}
Suppose $g$ is admissible and $h, k > 0$ satisfy
\begin{equation}
\label{eqn:newhk1}
k \leq 2 \pi (\log 2)^{-1/2} M_2 h^{1/2}.
\end{equation}
Then for all $y \in \mathbb{R}$, we have
\begin{equation}
\label{eqn:discgaussl1}
\biggl| 1 -  k \| G(\cdot,y) \|_{\ell^1} \biggr| \leq 4 \exp{ \left( -\frac{2 \pi^2 g^2(y) h}{k^2} \right) }.
\end{equation}
\end{lemma} \begin{proof}
Let
\begin{equation}
\label{eqn:discgaussdef}
\phi(x) = \frac{1}{\sqrt{2 \pi \sigma^2}} \exp{\left(-\frac{(x-\mu)^2}{2 \sigma^2}\right)}.
\end{equation}
Note that $G$ and $\phi$ coincide when $\mu = y + f(y) h$ and $\sigma^2 = g^2(y) h$.  For any $d > 0$, on the strip $S_d$, $\phi$ satisfies the hypotheses of Theorem \ref{thm:trap}.  We restrict attention to those $d$ satisfying 
$d > (2 \pi)^{-1} k \log 2$, so that
$(\sinh(\pi d/k))^{-1} \leq 4 \exp{\left( -\pi d/k \right)}$.  Then
$$
\int_{x \in \mathbb{R}} \left| \frac{1}{\sqrt{2 \pi \sigma^2}} \exp{\left(-\frac{(x + \irm d-\mu)^2}{2 \sigma^2}\right)} \right| dx = e^{{d^2}/(2\sigma^2)}.
$$
As the right-hand side does not change when we replace $d$ by $-d$, we have $\mathcal{N}(\phi,S_d) = 2 \exp ( d^2 / (2 \sigma^2) )$.  Using Theorem \ref{thm:trap} and $\int_{\mathbb{R}} \phi(x) \, dx = 1$, 
\begin{align}
\biggl| 1  - k \sum_{j=-\infty}^\infty \phi(j k) \biggr| 
& \leq \frac{\exp{\left(d^2/(2 \sigma^2)\right)}}{\sinh(\pi d/k)} \exp{\left(-\frac{\pi d}{k} \right)} \nonumber \\
\label{eqn:newhk3}
& \leq 4\exp{\left(\frac{d^2}{2\sigma^2}-\frac{2\pi d}{k}\right)}.
\end{align}
When $\sigma^2 = g(y) h$, we know by (\ref{eqn:gbound}) and (\ref{eqn:newhk1}) that
$d_\ast = 2 \pi \sigma^2/k \geq (2 \pi)^{-1} k \log 2$,
so we can choose $d = d_\ast$, the minimizer of (\ref{eqn:newhk3}) with respect to $d$, and maintain consistency.  Making this substitution and setting $\sigma^2 = g^2(y) h$, we have (\ref{eqn:discgaussl1}).   \end{proof}
For each $t_n$, we think of $\{ \hat{p}(x_j,t_n) \}_{j=-\infty}^\infty$ as an infinite sequence.  It is important to estimate the $\ell^1$ norm of this sequence.
\begin{lemma}
\label{lem:l1est}
Suppose $g$ is admissible and $h, k > 0$ satisfy (\ref{eqn:newhk1}).  Then for $n \geq 1$,
\begin{multline}
\label{eqn:ell1estimate}
( 1 - 4 \exp(-2 \pi^2 M_2^2 h / k^2))^{n-1} \leq 
\| \hat{p} (\cdot, t_n) \|_{\ell^1} / \| \hat{p}(\cdot,t_1) \|_{\ell^1} \\
\leq ( 1 + 4 \exp(-2 \pi^2 M_2^2 h / k^2))^{n-1},
\end{multline}
and the series defined by (\ref{eqn:phatn1}) converges uniformly.
\end{lemma} \begin{proof}
We prove this by induction with the base case of $n=1$, for which (\ref{eqn:ell1estimate}) is trivial.  Consider the infinite series on the right-hand side of (\ref{eqn:phatn1}) for $n = 1$ and fixed $h$ and $k$.  Using (\ref{eqn:gbound}), we have the elementary bound $0 \leq G(x,y) \leq (2 \pi M_2^2 h)^{-1/2}$.  Note that (\ref{eqn:nonsingic}) and (\ref{eqn:discgaussl1}) together give us an $\ell^1$ bound on $\{ \hat{p}(jk,t_1) \}_{j=-\infty}^\infty$.  Combining these two bounds, it is clear that (\ref{eqn:phatn1}) converges uniformly for $n=1$, i.e., $\hat{p}(y,t_2)$ converges uniformly.  

Now assume for fixed $n \geq 1$ that (\ref{eqn:ell1estimate}) holds, $\hat{p}(y,t_n) \geq 0$, $\| \hat{p}(\cdot, t_n) \|_{\ell^1}$ is finite, and $\hat{p}(y,t_{n+1})$ converges uniformly.  We now show that these properties hold with $n$ incremented by $1$.  By the induction hypotheses, we see that all terms of the convergent series on the right-hand side of (\ref{eqn:phatn1}) are nonnegative.  Hence $\hat{p}(y,t_{n+1}) \geq 0$.  We evaluate (\ref{eqn:phatn1}) at $x=x_i$:
\begin{equation}
\label{eqn:phatn1xi}
\hat{p}(x_i,t_{n+1}) = k \sum_{j=-\infty}^\infty G(x_i,y_j) \hat{p}(y_j,t_n).
\end{equation}
We take absolute values, sum over all $i$, and interchange the order of summation; this is all justified because all terms are nonnegative.  We obtain
$$
\| \hat{p}(\cdot, t_{n+1}) \|_{\ell^1} = \sum_{j=-\infty}^\infty \left[ k \sum_{i=-\infty}^\infty G(x_i,y_j) \right] \hat{p}(y_j,t_n).
$$
Applying (\ref{eqn:discgaussl1}) and (\ref{eqn:gbound}), we have
\begin{multline}
\label{eqn:ell1estimate_step}
( 1 - 4 \exp(-2 \pi^2 M_2^2 h / k^2)) \| \hat{p}(\cdot, t_n) \|_{\ell^1} \leq
\| \hat{p}(\cdot, t_{n+1}) \|_{\ell^1} \\ \leq  
( 1 + 4 \exp(-2 \pi^2 M_2^2 h / k^2)) \| \hat{p}(\cdot, t_n) \|_{\ell^1}.
\end{multline}
This shows that $\| \hat{p}(\cdot,t_{n+1}) \|_{\ell^1} < \infty$.  Now we return to the right-hand side of (\ref{eqn:phatn1}) with $n$ replaced by $n+1$.  Combining our elementary bound on $G$ with the $\ell^1$ bound on $\hat{p}(\cdot,t_{n+1})$, it is clear that the series converges uniformly.  From (\ref{eqn:ell1estimate_step}) we obtain upper and lower bounds for $\| \hat{p}(\cdot, t_{n+1}) \|_{\ell^1} / \| \hat{p}(\cdot, t_n) \|_{\ell^1}$.  Multiplying appropriately by (\ref{eqn:ell1estimate}), we advance $n$ by $1$.  \end{proof}
One consequence of Lemma \ref{lem:l1est} is that it enables us to give asymptotic conditions on $h$ and $k$ such that $\hat{p}$ is normalized correctly.
\begin{lemma}
\label{lem:l1normalization}
Suppose, in addition to the hypotheses of Lemmas \ref{lem:discgauss} and \ref{lem:l1est}, that $k = r_1 h^\rho$ for constants $r_1 > 0$ and $\rho > 1/2$.  Assume that $N = T/h$ for some fixed $T > 0$.  Then for $1 \leq n \leq N+1$,
\begin{equation}
\label{eqn:asympl1}
\lim_{h \to 0} k \| \hat{p}(\cdot,t_{n}) \|_{\ell^1} = 1.
\end{equation}
\end{lemma}\begin{proof}
Applying the hypotheses to the exponential terms in (\ref{eqn:ell1estimate}) with $n = N = T/h$, we have
\begin{equation}
\label{eqn:sumlimit}
\lim_{h \to 0} \left( 1 \pm 4 \exp{\left( -2\pi^2 M_2^2 r_1^{-2} h^{-2 \rho + 1} \right )} \right)^{T/h} = 1.
\end{equation}
Consequently, for any $n \in \{0, 1, \ldots, N\}$, we have
\begin{equation}
\label{eqn:limratio}
\lim_{h \to 0} \| \hat{p} (\cdot, t_{n+1}) \|_{\ell^1} / \| \hat{p}(\cdot,t_1) \|_{\ell^1} = 1.
\end{equation}
From (\ref{eqn:nonsingic}) and (\ref{eqn:discgaussl1}), we conclude that $k \| \hat{p}(\cdot,t_1) \|_{\ell^1} \to 1$ as $k \to 0$.  Then (\ref{eqn:asympl1}) follows immediately from (\ref{eqn:limratio}).  \end{proof}

\begin{lemma}
\label{lem:Gbound}
Suppose $f$ and $g$ are admissible and that
$$
a < \min \{d, M_2^2/(2 M_3 M_4) \}.
$$
Then for any $x, y \in \mathbb{R}$, there exist $A_2 > 0$ and $A_1, A_0 \in \mathbb{R}$ such that
\begin{equation}
\label{eqn:absgintermed}
|G(x,y+ \irm a)| = \frac{1}{\sqrt{2 \pi h |g(y + \irm a)|^2}} \exp
\left( -\frac{A_2 x^2 + A_1 x + A_0}{4 |g(y+\irm a)|^4 h} \right),
\end{equation}
and there exists $\gamma_0 \in (0,2)$ such that
$$
|G(x,y + \irm a)| \leq \frac{1}{\sqrt{2 \pi h M_2^2}} \exp \left( \frac{ a^2 (1 + h M_1)^2 }{h \gamma_0 M_2^2 } \right).
$$
\end{lemma} \begin{proof} 
We obtain (\ref{eqn:absgintermed}) by direct calculation of $|G(x,y+ \irm a)|$. The coefficients $A_2$, $A_1$, and $A_0$ are defined by
\begin{subequations}
\label{eqn:acoeffs}
\begin{align}
A_2 &= g^2(y - \irm a) + \text{c.c.} \\
A_1 &= -2 g^2(y - \irm a)(y + \irm a + f (y + \irm a) h) + \text{c.c.} \\
A_0 &= g^2(y - \irm a)(y^2 - a^2 + f^2(y + \irm a) h^2 \nonumber \\
 &\quad + 2 y \irm a + 2 (y + \irm a) f(y + \irm a) h) + \text{c.c.}
\end{align}
\end{subequations}
By ``c.c.'' we mean the complex conjugate of all preceding terms.  We have used the fact that because $f$ and $g$ are analytic on $S_d$, and because they are real-valued when restricted to the real axis, both $f$ and $g$ commute with complex conjugation.  That is, $\overline{f(y + \irm a)} = f(y - \irm a)$ and similarly for $g$ and $g^2$.  The upshot is that $A_2$, $A_1$, and $A_0$ are all real.

Let us now prove that $A_2 > 0$.  Define the function
$$
\theta(y,\epsilon) = g^2(y-\irm \epsilon)+g^2(y+\irm \epsilon),
$$
for $\epsilon \in [0,d)$.  For each fixed $y$, by the mean-value theorem, there exists $\xi$ such that
$$
\theta(y,\epsilon) - \theta(y,0) = \epsilon \frac{\partial \theta}{\partial \epsilon}(y,\xi).
$$
Note that $\xi$ may depend on $\epsilon$ and $y$.  Now we use (\ref{eqn:assumptions}) to compute
\begin{equation}
\label{eqn:usingm3}
\sup_{y \in \mathbb{R}, \epsilon \in (-d,d)} \left| \frac{\partial
    \theta}{\partial \epsilon} \right| = 4 \! \! \! \sup_{\substack{y \in \mathbb{R}\\ \epsilon \in (-d,d)}} \left| \Im(g(y+\irm \epsilon)g'(y + \irm \epsilon)) \right| \leq 4 M_3 M_4.
\end{equation}
Then using the previous two equations together with (\ref{eqn:gbound}), we have
\begin{equation}
\label{eqn:thetalb}
\theta(y,\epsilon) \geq \theta(y,0) - 4 \epsilon M_3 M_4 \geq 2 M_2^2 - 4 \epsilon M_3 M_4.
\end{equation}
The right-hand side will be positive as long as $\epsilon < \min\{d, M_2^2 / (2 M_3 M_4)\}$.  Given the hypothesis on $a$ in the statement of the lemma, $\theta(y,a) = A_2$ will be positive.  Because $A_2 > 0$, we can maximize the right-hand side of (\ref{eqn:absgintermed}) as a function of $x$; the global maximum
occurs at $x = -A_1/(2 A_2)$.  Then we have
$$
|G(x,y + \irm a)| \leq \frac{1}{\sqrt{2 \pi h M_2^2}} \exp \left( \frac{(2 a + \irm h (f(y - \irm a) - f(y + \irm a)) )^2}{4 h (g^2(y + \irm a) + g^2(y - \irm a))} \right).
$$
We suppose that $a = b M_2^2/(2 M_3 M_4)$ for some $b \in (0,1)$ such that $a < d$.  Then the lower bound (\ref{eqn:thetalb}) implies $\theta(y,a) \geq 2M_2^2 (1 - b)$.  We define $\gamma_0 = 2(1-b) \in (0,2)$ and write
\begin{multline}
\label{eqn:Gbound_almost}
|G(x,y + \irm a)| \leq \frac{1}{\sqrt{2 \pi h M_2^2}} \\
\times \exp \left( \frac{(2 a + \irm h (f(y - \irm a) - f(y + \irm a)) )^2}{h \gamma_0 M_2^2 } \right).
\end{multline}
Let $\Gamma$ be the segment connecting $y - \irm a$ to $y + \irm a$.  Note that $a < d$ implies that $\Gamma$ is completely contained in the strip $S_d$ where $f$ is analytic.  Using (\ref{eqn:lipschitz}), we have
\begin{align*}
| 2 a + \irm h(f(y - \irm a) - f&(y + \irm a)) | \\
 &\leq 2 | a | + h |f(y + \irm a) - f(y - \irm a)| \\
 &\leq 2 |a| + h \left| \oint_\Gamma f'(z) \, dz \right| \\
 &\leq 2 |a| + h \oint_\Gamma |f'(z)| \, |dz| \\
 &\leq 2 |a| (1 + h M_1)
\end{align*}
Using this estimate in (\ref{eqn:Gbound_almost}) finishes the proof.  \end{proof}

\begin{lemma}
\label{lem:membership}
Suppose that $f$ and $g$ are admissible, that $h, k > 0$ satisfy (\ref{eqn:newhk1}), and that $a < \min\{d, M_2^2/(2 M_3 M_4)\}$.  Then the integrand (\ref{eqn:ourintegrand}), considered as a function of $y$, is a member of $B(S_a)$, i.e., $\varphi(x,\cdot,t_n) \in B(S_a)$.
\end{lemma} \begin{proof}
There are three conditions for membership in $B(S_a)$, which we verify in turn.  First, we check that $\varphi$ is analytic on $S_a$.  At time step $t_1$, we have $\hat{p}(y,t_1) = G(y,C)$, the analyticity of which follows from (\ref{eqn:gnz}) and the lower bound in (\ref{eqn:gbound}).  The arguments made earlier regarding the convergence of (\ref{eqn:phatn1xi}) hold equally well with $x_i$ replaced by any $x$.  This implies that for $n \geq 1$,  $\hat{p}(y,t_{n+1})$ is analytic in $y$ on $S_d$, so the integrand $\varphi$ is analytic on $S_a \subset S_d$.  
Next, we consider
\begin{equation}
\label{eqn:bigphi}
\Phi(x,y,t_n) = \int_{b=-a}^{a} |\varphi(x,y + \irm b,t_n)| \, db.
\end{equation}
Let $z_j = jk$.  Since
\begin{equation}
\label{eqn:phiyplusia}
\hat{p}(y + \irm a,t_{n+1}) = k \sum_{j=-\infty}^{\infty}G(y + \irm a,z_j)\hat{p}(z_j,t_n),
\end{equation}
we have
\begin{align*}
&\Phi(x,y,t_{n+1}) \\
&\leq k \sum_{j=-\infty}^{\infty}\hat{p}(z_j,t_n)\int_{b=-a}^{a}|G(y+  \irm b, z_j)| |G(x, y + \irm b)| \, db  \\
&= k \sum_{j=-\infty}^{\infty}\hat{p}(z_j,t_n) G(y,z_j) \\
& \qquad \times  \int_{b=-a}^{a} \exp{\left(\dfrac{b^2}{2g^2(z_j)h}\right)} |G(x,y+
  \irm b)| \, db \\
& \leq  \frac{1}{\sqrt{2 \pi h M_2^2}}  \\
& \qquad \times \int_{b=-a}^{a} \exp{\left(\dfrac{b^2}{2 M_2^2 h}\right)}  \exp \left( \frac{ b^2 (1 + h M_1)^2 }{h \gamma_0 M_2^2 } \right) \, db \\
& \qquad \times \underbrace{k \sum_{j=-\infty}^{\infty}\hat{p}(z_j,t_n) G(y,z_j)}_{\hat{p}(y, t_{n+1})}.
\end{align*}
To derive the last inequality, we have applied Lemma \ref{lem:Gbound} and (\ref{eqn:gbound}).  By Lemma \ref{lem:l1est}, we know $\hat{p}(y, t_{n+1})$ converges uniformly.  We integrate with respect to $y$, bring the integral into the sum, and use (\ref{eqn:Gint}).  In this way, we derive $\| \hat{p}(\cdot, t_{n+1}) \|_1 = k \| \hat{p}(\cdot, t_n) \|_{\ell^1} < \infty$.  Therefore, $\hat{p}(y, t_{n+1}) \to 0$ as $|y| \to \infty$; in the same limit, we have $\Phi(x,y,t_{n+1}) = O(|y|^\alpha)$ for $\alpha = 0$, satisfying (\ref{eqn:bsddecay}).

Next, we establish a bounded, real function $L_n$ such that for each $x \in \mathbb{R}$,
\begin{multline}
\label{eqn:condtn1}
\mathcal{N} := \int_{\mathbb{R}} |G(x,y + \irm a) \hat{p}(y + \irm a,t_n)| \, dy \\
+ \int_{\mathbb{R}} |G(x,y - \irm a) \hat{p}(y - \irm a,t_n)| \, dy 
\leq L_n(x).
\end{multline}
We need this estimate in order to apply Theorem
\ref{thm:trap}. For this purpose, we seek an upper bound on $\mathcal{N}$ that does not depend essentially on the spatial discretization parameter $k$. 
Starting again from (\ref{eqn:phiyplusia}), we have
\begin{align}
&\int_{y \in \mathbb{R}} |G(x,y + \irm a) \hat{p}(y + \irm a,t_{n+1})| \, dy \nonumber \\
&\leq k \sum_{j=-\infty}^{\infty}\hat{p}(z_j,t_n) % \nonumber \\
%&\qquad 
\int_{\mathbb{R}}|G(y+  \irm a, z_j)| |G(x, y + \irm a)| \, dy \nonumber \\
&= k \sum_{j=-\infty}^{\infty}\hat{p}(z_j,t_n) \nonumber \\
&\qquad \int_{\mathbb{R}} \exp{\left(\dfrac{a^2}{2g^2(z_j)h}\right)}G(y,z_j)|G(x,y+ \irm a)| \, dy \nonumber \\
\begin{split}
\label{eqn:condtn1begin}
&\leq k \exp{\left(\dfrac{a^2}{2M_2^2h}\right)} \\
& \quad \times  \sum_{j=-\infty}^{\infty}\hat{p}(z_j,t_n) 
 \int_{\mathbb{R}} G(y,z_j)|G(x,y+ \irm a)| \, dy
\end{split}  \\
%& \leq k \exp{\left(\dfrac{a^2}{2M_2^2h}\right)}\| \hat{p}(\cdot, t_n)\|_{{\ell^1}} \nonumber \\
%& \qquad \sup_{z}\left[\int_{\mathbb{R}}G(y,z)|G(x,y+ \irm a)| \, dy\right] \nonumber \\
\label{eqn:condtn1intermed}
& \leq k \exp{\left(\dfrac{a^2}{2M_2^2h}\right)}\| \hat{p}(\cdot, t_n)\|_{{\ell^1}} \psi(x,a),
\end{align}
where 
\begin{equation}
\label{eqn:psidef}
\psi(x,a) = \sup_{z \in \mathbb{R}}\left[\int_{y \in \mathbb{R}}G(y,z)|G(x,y+ \irm a)|dy\right].
\end{equation}
Examining (\ref{eqn:absgintermed}), we see that the right-hand side of (\ref{eqn:condtn1intermed}) is invariant under the reflection $a \mapsto (-a)$.  We define the real-valued function
$$
L_{n+1}(x) = 2 k \exp{\left(\dfrac{a^2}{2M_2^2h}\right)}\| \hat{p}(\cdot, t_n)\|_{{\ell^1}} \psi(x,a),
$$
and note that (\ref{eqn:condtn1intermed}) implies $\mathcal{N} \leq
L_n(x)$, as required by (\ref{eqn:condtn1}).
Our task now is to demonstrate that $L_n$ is finite.  By Lemma
\ref{lem:Gbound} and (\ref{eqn:Gint}), we have
$$
\psi(x,a) \leq \frac{1}{\sqrt{2 \pi h M_2^2}} \exp \left( \frac{ a^2 (1 + h M_1)^2 }{h \gamma_0 M_2^2 } \right).
$$
%\begin{align*}
%\psi(x,a) &\leq \sup_z \Biggl[ \frac{1}{\sqrt{2 \pi h M_2^2}} \exp
%            \left( \frac{ a^2 (1 + h M_1)^2 }{h \gamma_0 M_2^2 }
%            \right) \\ &\qquad \qquad \int_{y \in \mathbb{R}} G(y,z) \, dy \Biggr] \\
%&\leq  \frac{1}{\sqrt{2 \pi h M_2^2}} \exp \left( \frac{ a^2 (1 + h M_1)^2 }{h \gamma_0 M_2^2 } \right).
%\end{align*}
Using this estimate in (\ref{eqn:condtn1intermed}), we obtain
$$
L_{n+1}(x) \leq 2 k \exp{\left(\dfrac{a^2}{2M_2^2h}\right)}\| \hat{p}(\cdot, t_n)\|_{{\ell^1}} 
\frac{1}{\sqrt{2 \pi h M_2^2}} \exp \left( \frac{ a^2 (1 + h M_1)^2 }{h \gamma_0 M_2^2 } \right).
$$
Note that the bound on the right-hand side does not depend on $x$ at all.  The dependence on $k$ is confined to the terms $k \| \hat{p}(\cdot,t_{n}) \|$.  By Lemmas \ref{lem:discgauss} and \ref{lem:l1est} together with (\ref{eqn:discgaussl1}), 
\begin{multline*}
k \| \hat{p}(\cdot,t_n) \| \leq \left( 1 + 4 \exp(-2 \pi^2 g^2(C)
  h/k^2) \right) \\ \times \left(1 + 4 \exp (-2 \pi^2 M_2^2 h/k^2) \right)^{n-1}
\leq 5^n < \infty
\end{multline*}
for all $k \geq 0$.  In sum, we have shown that for fixed $h > 0$, fixed $n \geq 1$, and $a < \min\{d, M_2^2/(2 M_3 M_4)\}$, $L_n(x)$ is bounded uniformly in $x$ and $k$.  We have demonstrated that (\ref{eqn:condtn1}) holds.  Thus $\varphi(x,\cdot,t_n) \in B(S_a)$.   \end{proof}

\section{Convergence Theorem}
\label{sect:convthm}
Let
\begin{equation}
\label{eqn:errordef}
E(y,t_n) = \tilde{p}(y,t_n) - \hat{p}(y, t_n).
\end{equation}
In this section, we establish conditions under which $\| E(\cdot,T)
\|_1$ goes to zero at an exponential rate.
\begin{theorem}
\label{thm:convergence}
Suppose that $f$ and $g$ are admissible in the sense of Definition
\ref{def:admissible}.  Assume that 
\begin{equation}
\label{eqn:assumption2}
k = r_1 h^\rho
\end{equation}
for constants $r_1 > 0$ and $\rho > 1/2$.  Choose $a < \min\{d, M_2^2/(2 M_3 M_4)\}$ such that
\begin{equation}
\label{eqn:assumption3}
a = r_2  h^{1/2}
\end{equation}
for some $r_2 > 0$.  Assume that $h, k$ satisfy (\ref{eqn:newhk1}) and that $k < 2 \pi a/\log 2$.  For fixed $T > 0$, choose
\begin{equation}
\label{eqn:assumption4}
h \in \left(0, \min \{T, (M_2^2/(4 M_3 M_4 r_2))^2\} \right)
\end{equation}
such that $N = T/h \in \mathbb{N}^+$.  To be clear, $r_1$ and $r_2$ are constants that do not depend on $h$.  Then
\begin{equation}
\label{eqn:finerr}
\| E(\cdot,T) \|_1  \leq
c_\star h^{-1} \exp(-2 \pi r_2 r_1^{-1} h^{1/2 - \rho}) (1 + o(h) + o(k))
\end{equation}
where $o(h)$ and $o(k)$ stand for terms that vanish as $h \to 0$ and $k \to 0$, and $c_\star > 0$ is a constant that does not depend on $h$.
\end{theorem}
\begin{proof}
We begin with
\begin{align*}
\tilde{p}(x,t_{n+1}) &= \int_{y \in \mathbb{R}} G(x,y) \tilde{p}(y,t_n) \, dy\\
&= \int_{y \in \mathbb{R}}G(x,y)\hat{p}(y, t_n) \, dy +\int_{y \in \mathbb{R}} G(x,y)E(y, t_n) \, dy.
\end{align*}
We now apply the trapezoidal rule to the first integral.  For each $x$ and $t_n$, we let $\tau(x,t_n)$ denote the quadrature error incurred, i.e.,
\begin{align}
\int_{y \in \mathbb{R}} G(x,y) \hat{p}(y,t_n) \, dy &= k \sum_{j=-\infty}^\infty G(x,y_j) \hat{p}(y_j,t_n) + \tau(x,t_n) \nonumber \\
\label{eqn:taudef}
 &= \hat{p}(x,t_{n+1}) + \tau(x,t_n).
\end{align}
We use this in the previous equation to derive
$$
E(x,t_{n+1}) = \int_{y \in \mathbb{R}} G(x,y) E(y,t_n) \, dy + \tau(x,t_n).
$$
Taking absolute values, we apply the triangle inequality together with $G \geq 0$ to obtain
$$
|E(x,t_{n+1})| \leq \int_{y \in \mathbb{R}} G(x,y) |E(y,t_n)| \, dy + |\tau(x,t_n)|.
$$
Integrating over $x$ and using Fubini's theorem and (\ref{eqn:Gint}), 
\begin{equation}
\label{error1}
\|E(\cdot,t_{n+1})\|_1 - \|E(\cdot,t_n)\|_1 \leq \| \tau(\cdot,t_n) \|_1.
\end{equation}
Summing both sides from $n=1$ to $n=N-1$ and using (\ref{eqn:nonsingic}),
\begin{equation}
\label{eqn:e1norm}
\|E(\cdot,T)\|_1 \leq \sum_{n=1}^{N-1} \| \tau(\cdot,t_n) \|_1.
\end{equation}
We apply Lemma \ref{lem:membership} and Theorem \ref{thm:trap} to produce
\begin{equation}
\label{eqn:tauupper}
| \tau(x,t_n) | \leq \frac{\mathcal{N}}{2 \sinh(\pi a/k)} \exp(- \pi a/k),
\end{equation}
where $\tau$ and $\mathcal{N}$ are defined by (\ref{eqn:taudef}) and (\ref{eqn:condtn1}), respectively.  Combining (\ref{eqn:condtn1begin}) with (\ref{eqn:absgintermed}), we have
\begin{multline*}
\int_{y \in \mathbb{R}} | G(x,y + \irm a) \hat{p}(y + \irm a, t_{n+1} ) | \, dy \\
\leq \exp \left( \frac{a^2}{2 M_2^2 h} \right) 
 k \sum_{j=-\infty}^\infty \hat{p}(z_j, t_n) \\
\times \int_{y \in \mathbb{R}} \frac{G(y,z_j)}{\sqrt{2 \pi h |g(y + \irm a)|^2}} \\ \times \exp \left( -\frac{A_2 x^2 + A_1 x + A_0}{4 |g(y+\irm a)|^4 h} \right) \, dy,
\end{multline*}
where again $A_2$, $A_1$, and $A_0$ are defined by
(\ref{eqn:acoeffs}).  We see that the right-hand side of this
inequality is invariant under $a \mapsto -a$, and so we write
\begin{multline*}
\mathcal{N} 
\leq 2 \exp \left( \frac{a^2}{2 M_2^2 h} \right)  k \sum_{j=-\infty}^\infty \hat{p}(z_j, t_n) \\
\times \int_{y \in \mathbb{R}} \frac{G(y,z_j)}{\sqrt{2 \pi h |g(y + \irm a)|^2}} \\ \times \exp \left( -\frac{A_2 x^2 + A_1 x + A_0}{4 |g(y+\irm a)|^4 h} \right) \, dy.
\end{multline*}
For $a < \min\{d, M_2^2/(2 M_3 M_4)\}$, we have shown that the coefficient $A_2$ is positive on $S_a$.  This enables us to integrate both sides with respect to $x$:
\begin{multline*}
\int_{x \in \mathbb{R}} \mathcal{N} \, dx \leq 2 \sqrt{2} \exp \left( \frac{a^2}{2 M_2^2 h} \right) k \sum_{j=-\infty}^\infty \hat{p}(z_j, t_n) \\
\times \int_{y \in \mathbb{R}}
\frac{ G(y,z_j) |g(y + \irm a)| }{ \sqrt{ g^2(y + \irm a) + g^2(y - \irm a) } } \\
\times \exp \left( \frac{(2 a + \irm h (f(y - \irm a) - f(y + \irm a)) )^2}{4 h (g^2(y + \irm a) + g^2(y - \irm a))} \right) \, dy.
\end{multline*}
On the right-hand side, we have carried out the $x$ integral first;
the changing of the order of summation and integration is justified by
the nonnegativity of every term.  Next, we apply estimates established
via Lemma \ref{lem:Gbound}.  We obtain
\begin{multline*}
\int_{x \in \mathbb{R}} \mathcal{N} \, dx \leq 2 \sqrt{2} \exp \left( \frac{a^2}{2 M_2^2 h} \right) \frac{M_3}{\gamma_0^{1/2} M_2} \\
\times \exp \left( \frac{a^2 (1 + h M_1)^2}{ h \gamma_0 M_2^2} \right) k \sum_{j=-\infty}^\infty \hat{p}(z_j,t_n) 
\end{multline*}
Combining this with (\ref{eqn:tauupper}) and the estimates from Lemma \ref{lem:discgauss}, we have
\begin{multline*}
\int_{x \in \mathbb{R}} |\tau(x,t_n)| \, dx \leq 
4 \sqrt{2} \exp \left( \frac{a^2}{2 M_2^2 h} \right) \frac{M_3}{\gamma_0^{1/2} M_2} \\ \times \exp \left( \frac{a^2 (1 + h M_1)^2}{ h \gamma_0 M_2^2} \right) \exp(-2 \pi a/k) k \| \hat{p}(\cdot, t_n) \|_{\ell^1}.
\end{multline*}
Using (\ref{eqn:ell1estimate}), we obtain
\begin{multline*}
\| \tau(\cdot, t_n) \|_1 \leq 
4 \sqrt{2} M_3 \gamma_0^{-1/2} M_2^{-1} \\
\times \exp \left( \frac{a^2}{2 M_2^2 h} \right) \exp \left( \frac{a^2 (1 + h M_1)^2}{ h \gamma_0 M_2^2} \right) \exp(-2 \pi a/k) \\
\times k \| \hat{p}(\cdot,t_1) \|_{\ell^1} (1 + 4 \exp(-2 \pi^2 M_2^2 h/k^2))^{n-1}.
\end{multline*}
We sum both sides from $n=1$ to $n=N-1$:
\begin{multline}
\label{eqn:tauestimate}
\sum_{n=1}^{N-1} \| \tau(\cdot, t_n) \|_1 \leq
\sqrt{2} M_3 \gamma_0^{-1/2} M_2^{-1} \\
\times \exp \left( \frac{a^2}{2 M_2^2 h} \right) \exp \left( \frac{a^2 (1 + h M_1)^2}{ h \gamma_0 M_2^2} \right) \\
\times h^{-1} \exp(-2 \pi a/k) k \| \hat{p}(\cdot,t_1) \|_{\ell^1} \\
\times \biggl[ h \sum_{n=1}^{N-1} (1 + 4 \exp(-2 \pi^2 M_2^2 h/k^2))^{n-1} \biggr].
\end{multline}
We now use (\ref{eqn:e1norm}) and hypotheses (\ref{eqn:assumption2}) and
(\ref{eqn:assumption3}):
\begin{multline}
\label{eqn:tauestimate2}
\| E(\cdot,T) \|_1 \leq
\sqrt{2} M_3 \gamma_0^{-1/2} M_2^{-1} \\
\times \exp \left( \frac{r_2^2}{2 M_2^2} \right) \exp \left( \frac{r_2^2 (1 +
    h M_1)^2}{ \gamma_0 M_2^2} \right) T \\
\times h^{-1} \exp(-2 \pi r_2 r_1^{-1} h^{1/2 - \rho}) k \|
\hat{p}(\cdot,t_1) \|_{\ell^1} \\
\times \biggl[ \frac{h}{T} \sum_{n=1}^{N-1} (1 + 4
  \exp(-2 \pi^2 M_2^2 r_1^{-2} h^{1 - 2 \rho} ))^{n-1} \biggr].
\end{multline}
By (\ref{eqn:assumption4}), we have $h \leq T$.  By the definition of
$\gamma_0$ in Lemma \ref{lem:Gbound}, we have that $\gamma_0 =
2(1-b)$ where
$$
b = 2M_3 M_4 a/M_2^2 = 2 M_3 M_4 r_2 h^{1/2} / M_2^2.
$$
Assumption (\ref{eqn:assumption4}) now implies that $b \leq 1/2$ and 
$\gamma_0^{-1} \leq 1$.  We write
$$
c_\star = \sqrt{2} M_3 M_2^{-1}
\exp \left( \frac{r_2^2}{2 M_2^2} \right) \exp \left( \frac{r_2^2 (1 +
    T M_1)^2}{ M_2^2} \right) T.
$$
Let $\xi(h) = 4 \exp( -c_1 h^{-c_2} )$,
where $c_1$ and $c_2$ are positive constants with no dependence on
$h$. We check that $\xi$ satisfies the hypotheses of Lemma \ref{lem:o1t}; $h^{-\gamma} \xi(h)$ has a global maximum at
$h_\ast = (c_1 c_2 / \gamma)^{1/c_2}$, and so we have $\xi(h) \leq
\epsilon h^\gamma$ for $\epsilon = h_\ast^{-\gamma} \xi(h_\ast)$, any
choice of $\gamma > 1$, and all $h > 0$.   With $c_1 = 2 \pi^2
\gamma^2$ and $c_2 = 2 \rho - 1$, we apply Lemma \ref{lem:o1t} to the
term in square brackets on the right-hand side of
(\ref{eqn:tauestimate2}).  We conclude that
$$
\frac{h}{T} \sum_{n=1}^{N-1} (1 + 4
  \exp(-2 \pi^2 M_2^2 r_1^{-2} h^{1 - 2 \rho} ))^{n-1} = 1 + o(h)
$$
as $h \to 0$ with $N = T/h$.  By Lemma \ref{lem:discgauss},
$k \| \hat{p}(\cdot,t_1) \|_{\ell^1} = 1 + o(k)$
as $k \to 0$.  Putting everything together, we are left with
(\ref{eqn:finerr}).   \end{proof}

We are now in a position to combine our result with an earlier result
from the literature to establish the convergence
of $\hat{p}$ to $p$.

\begin{corollary}
\label{cor:phatp}
In addition to the hypotheses of Theorem \ref{thm:convergence},
suppose there exist constants $\overline{f}_k, \overline{g}_k > 0$ such that
$\sup_{x \in \mathbb{R}} |f^{(k)}(x)| \leq \overline{f}_k$ and 
$\sup_{x \in \mathbb{R}} |g^{(k)}(x)| \leq \overline{g}_k$
for all $k \geq 0$.  Note that for $k=1$, the first condition is redundant with
(\ref{eqn:lipschitz}); for $k=0$ and $k=1$, the second condition is
redundant with (\ref{eqn:gbound}) and (\ref{eqn:glip}).  Then we have
$$
\| p(\cdot,T) - \hat{p}(\cdot,T) \|_1 = O(h)
$$
\end{corollary}
\begin{proof}
We have
\begin{equation}
\label{eqn:usetri}
\| p(\cdot,T) - \hat{p}(\cdot,T) \|_1 \leq 
\| p(\cdot,T) - \tilde{p}(\cdot,T) \|_1 + \| \tilde{p}(\cdot,T) -
\hat{p}(\cdot,T) \|_1
\end{equation}
To handle the first term, we appeal to
Corollary 2.1 from \citet{BallyTalay1996}.  Our lower bound on $g$ in
(\ref{eqn:gbound}) corresponds to Bally and Talay's uniform
ellipticity hypothesis ``H1''; we may then apply Equations (27-28) from
\citet{BallyTalay1996} to derive
$$
| p(x,T) - \tilde{p}(x,T) | \leq h \mathcal{H}_1 \exp
\left( - \mathcal{H}_2 x^2/T \right)
$$
for constants $\mathcal{H}_1, \mathcal{H}_2 > 0$ that do not depend on
$h$.  Therefore,
$$
\| p(\cdot,T) - \tilde{p}(\cdot,T) \|_1 \leq h \mathcal{H}_1 \left(
  \frac{\pi T}{\mathcal{H}_2} \right)^{1/2} = O(h).
$$
Returning to (\ref{eqn:usetri}), by Theorem \ref{thm:convergence}, 
the second term on the right-hand side goes to zero much faster than $h$,
finishing the proof.
 \end{proof}

\section{Boundary Truncation}
\label{sect:trunc}
In practice, in lieu of the infinite sum (\ref{eqn:phatn1}), we compute approximate densities using the following truncation:
\begin{equation}
\label{eqn:pcircn1}
\mathring{p}(x,t_{n+1}) = k \sum_{j=-M}^M G(x,y_j) \mathring{p}(y_j,t_n)
\end{equation}
As in (\ref{eqn:nonsingic}), we take $\mathring{p}(x,t_1) = G(x,C)$
and use (\ref{eqn:pcircn1}) starting with $n=1$.  Let us denote the error due to truncation by
\begin{equation}
\label{eqn:rdef}
r(x,t_{n+1}) = \hat{p}(x,t_{n+1}) - \mathring{p}(x,t_{n+1})
\end{equation}
By (\ref{eqn:nonsingic}), we have $r(x,t_1) \equiv 0$.  For $n \geq
1$, we have
\begin{equation}
\label{eqn:rbreakdown}
r(x,t_{n+1}) = k \Biggl( \sum_{|j|>M} G(x,y_j) \hat{p}(y_j,t_n) 
 + \sum_{|j| \leq M} G(x,y_j) r(y_j,t_n) \Biggr).
\end{equation}
Based on the right-hand side, we see that it will be important to estimate the tail sum $\sum_{|j|>M} \hat{p}(x_j,t_n)$.  We accomplish this using a Chernoff bound.  To arrive at this bound, we construct a sequence of random variables $\{Q_n\}_{n \geq 1}$.  We first define a normalization constant at time $n$:
\begin{equation}
\label{eqn:kndef}
K_n = \| \hat{p}(\cdot,t_n) \|_{\ell^1} = \sum_i \hat{p}(x_i,t_n).
\end{equation}
By (\ref{eqn:ell1estimate}), we know that $K_n < \infty$ for $k > 0$
and $h > 0$.  Let
\begin{equation}
\label{eqn:qdef}
q(x_i,t_n) = \frac{\hat{p}(x_i,t_n)}{K_n},
\end{equation}
so that $\sum_i q(x_i,t_n) = 1$.  For each $n$, we postulate a random variable $Q_n$ with state space $\{k \mathbb{Z}\}$ and probability mass function $q(\cdot,t_n)$.   In order to apply a Chernoff bound to $Q_n$, we must estimate its moment generating function.
\begin{lemma}
\label{lem:chernoff}
Suppose $f$ and $g$ are admissible.  Suppose $k = h^\rho$ for some $\rho > 1/2$, and that $h, k > 0$ satisfy (\ref{eqn:newhk1}).  Then there exists $h_\ast$ such that for all $h \in [0, h_\ast)$, all $s \in \mathbb{R}$, and all $n$ satisfying $0 \leq n \leq (N-1)$,
$$
k E[e^{s Q_{n+1}}] < \frac{3}{2} \exp \left[ T \left( \frac{M_3^2 s^2}{2} + f(0) s \right) \right] \left( \frac{1}{2} + \exp (C s e^{M_1 T} ) \right)  < \infty.
$$
\end{lemma}
\begin{proof}
We begin with our estimate of the moment generating function of $Q_{n+1}$.  The calculation proceeds in two phases.  The first phase is exact; note that in what follows we use the notation $y_j = jk$, $w_j = y_j + f(y_j) h$, and $g^2 = g^2(y_j)$:
\begin{align}
&E[e^{s Q_{n+1}}] = \sum_{i=-\infty}^\infty e^{s x_i} q(x_i,t_{n+1}) \nonumber \\
 &= \frac{k}{K_{n+1}} \sum_i e^{s x_i} \sum_j \frac{1}{\sqrt{2 \pi g^2 h}} \exp \left( -\frac{(x_i - w_j)^2}{2 g^2 h} \right) \hat{p}(y_j,t_{n}) \nonumber \\
 &= \frac{k}{K_{n+1}} \sum_j \sum_i \frac{1}{\sqrt{2 \pi g^2 h}} \exp \left( -\frac{x_i^2 - 2 x_i w_j + w_j^2 - 2 g^2 h s x_i}{2 g^2 h} \right) \hat{p}(y_j,t_{n}) \nonumber \\
\label{eqn:mgf0}
 &= \frac{1}{K_{n+1}} \sum_j \zeta_s(j) \exp \left( -\frac{w_j^2 - (w_j + g^2 h s)^2}{2 g^2 h} \right) \hat{p}(y_j,t_{n}), 
\end{align}
where
$$
\zeta_s(j) = k \sum_i \frac{1}{\sqrt{2 \pi g^2 h}} \exp \left( -\frac{(x_i - (w_j + g^2 h s))^2}{2 g^2 h} \right).
$$
It is at this point that we begin to estimate.  Note that the summand is in fact a discrete Gaussian $\phi(x_i)$, as in (\ref{eqn:discgaussdef}), with $\mu = w_j + g^2(y_j) h s$ and $\sigma^2 = g^2(y_j) h$.  Hence we may apply the same reasoning from Lemma \ref{lem:discgauss} to obtain
\begin{equation}
\label{eqn:zetaestimate}
\zeta_s(j) \leq 1 + 4 \exp \left( \frac{-2 \pi^2 g^2(y_j) h}{k^2} \right) 
\leq 1 + 4 \exp \left( \frac{-2 \pi^2 M_2^2 h}{k^2} \right).
\end{equation}
Next, we turn our attention to the remaining exponential in (\ref{eqn:mgf0}).  We use (\ref{eqn:gbound}), the mean value theorem, and the definition of $w_j$ to obtain:
\begin{align}
\exp \left( -\frac{w_j^2 - (w_j + g^2 h s)^2}{2 g^2 h} \right) &= \exp \left( w_j s + \frac{1}{2} g^2(y_j) h s^2 \right) \nonumber \\
 &\leq e^{ M_3^2 h s^2/2 } \exp (y_j s + f(y_j) h s) \nonumber \\
 &\leq e^{ M_3^2 h s^2/2 } \exp (y_j s + f(0) h s + \beta y_j h s) \nonumber \\
\label{eqn:whatwecalledtau}
 &\leq e^{ M_3^2 h s^2/2 + f(0) h s } \exp ( y_j s (1 + \beta h) )
\end{align}
Here $\beta = f'(y)$ for some $y \in (0,y_j)$.  Now we combine (\ref{eqn:mgf0}), (\ref{eqn:zetaestimate}), and (\ref{eqn:whatwecalledtau}).  The result is
\begin{multline}
\label{eqn:mgf1}
E[e^{s Q_{n+1}}] \leq \frac{K_{n}}{K_{n+1}} \left( 1 + 4 \exp ( -2 \pi^2 M_2^2 h/k^2 ) \right)  e^{ M_3^2 h s^2/2 + f(0) h s } \\ \times \frac{1}{K_{n}} \sum_j \exp ( y_j s (1 + \beta h) ) \hat{p}(y_j,t_{n})
\end{multline}
We recognize the expression on the second line as the moment generating function of $Q_{n}$ evaluated at $s' = s(1 + \beta h)$.  Therefore,
\begin{align*}
%\label{eqn:mgf2}
k E[e^{s Q_{n+1}}] &\leq \frac{K_{n}}{K_{n+1}} \left( 1 + 4 \exp ( -2 \pi^2 M_2^2 h/k^2 ) \right)  \\
&\qquad \times e^{ M_3^2 h s^2/2 + f(0) h s } k E[e^{s (1 + \beta h) Q_{n}}] \\
%\label{eqn:mgf3}
 &\leq \underbrace{\frac{K_1}{K_{n+1}} \left( 1 + 4 \exp ( -2 \pi^2 M_2^2 h/k^2 ) \right)^n}_{\zeta_1(h)} \\
&\qquad \times e^{T(M_3^2 s^2/2 + f(0) s)} \underbrace{k E[e^{s (1 + \beta h)^n Q_1} ]}_{\zeta_2(h)}.
\end{align*}
The main question now is what happens as $h \to 0$ and $N \to \infty$ such that $h N = T$.  We assume that $0 \leq n \leq (N-1)$.  Because $k = r_1 h^\rho$ for $\rho > 1/2$, we know by Lemma \ref{lem:l1est} that $\zeta_1(h) \to 1$ as $h \to 0$.  Hence there exists $h_\ast^1$ such that $h \in [0,h_\ast^1)$ ensures that $|\zeta_1(h) - 1| < 1/2$, i.e., $\zeta_1(h) < 3/2$.  Next, consider 
\begin{align*}
\zeta_2(h) &= k E[e^{s (1 + \beta h)^n Q_1}] \\
 &= k \sum_{i=-\infty}^\infty e^{s (1 + \beta h)^n x_i} \hat{p}(x_i,t_1) \\
 &= k \sum_{i=-\infty}^\infty e^{s (1 + \beta h)^n x_i} G(x_i,C) \\
 &= \exp \Bigl( (C + f(C) h) s (1 + \beta h)^n +\frac{h g^2(C) s^2}{2} (1 + \beta h)^{2n} \Bigr) k \sum_{i=-\infty}^\infty \phi(x_i),
\end{align*}
where $\phi(x)$ is the Gaussian density defined in (\ref{eqn:discgaussdef}) with
\begin{align*}
\mu &= C + f(C) h + h g^2(C) s (1 + \beta h)^n \\
\sigma^2 &= h g^2(C)
\end{align*}
Now we apply Lemma \ref{lem:discgauss}, $n \leq (N-1)$, and (\ref{eqn:lipschitz}) to obtain
\begin{align*}
\zeta_2(h) &\leq \exp \Bigl( |C + f(C) h| s (1 + M_1 h)^N 
+ \frac{h g^2(C) s^2}{2} (1 + M_1 h)^{2N} \Bigr) \\
& \qquad \times (1 + 4 \exp (- 2 \pi^2 g^2(C) h / k^2)).
\end{align*}
As before, $h k^{-2} = r_1^{-2} h^{1 - 2 \rho} \to +\infty$ as $h \to 0$, and the term on the third line goes to $1$ as $h \to 0$.  Since $\lim_{h \to 0^+} (1 + M_1 h)^N = e^{M_1 T}$, we have
$$
\lim_{h \to 0^+} \zeta_2(h) \leq \exp \left( C s e^{M_1 T} \right).
$$
Thus there exists $h^2_\ast$ such that $h \in [0,h^2_*)$ implies
$$
\left| \zeta_2(h) - \exp \left( C s e^{M_1 T} \right) \right| \leq \frac{1}{2}.
$$
Taking $h_\ast = \min \{ h^1_\ast, h^2_\ast \}$ finishes the proof.
 \end{proof}

We can now give conditions under which $r$, defined in
(\ref{eqn:rdef}), converges to zero.
\begin{lemma}
\label{lem:rl1}
Suppose $f$ and $g$ are admissible in the sense of Definition
\ref{def:admissible}.  Suppose $k = h^\rho$ for $\rho > 1/2$,
and that $h, k > 0$ satisfy (\ref{eqn:newhk1}).
For $\varepsilon \geq 1$, let
\begin{equation}
\label{eqn:ourM}
M = \lceil (\varepsilon + \rho + 1) (-\log h) / k \rceil.
\end{equation}
Let $h_\ast$ be as in Lemma \ref{lem:chernoff}.
Then for $h < h_\ast$, $\displaystyle k \sum_{ |i| \leq M } |r(x_i,T)| = O(h)$.
\end{lemma}
\begin{proof}
We start with
$$
|r(x_i,t_{n+1})| \leq k \sum_{|j|>M} G(x_i,y_j) \hat{p}(y_j,t_n) 
 + k \sum_{|j| \leq M} G(x_i,y_j) |r(y_j,t_n)|.
 $$
Summing over $i$, we obtain
\begin{multline*}
\sum_{|i| \leq M} |r(x_i,t_{n+1})| \leq k \sum_{|j|>M} \sum_{|i| \leq M} G(x_i,y_j) \hat{p}(y_j,t_n) \\
 + k \sum_{|j| \leq M} \sum_{|i| \leq M} G(x_i,y_j) |r(y_j,t_n)|.
\end{multline*}
Using (\ref{eqn:discgaussl1}) together with (\ref{eqn:gbound}), we have
\begin{multline}
\label{eqn:ourform}
\sum_{|i| \leq M} |r(x_i,t_{n+1})| \\
 \leq  ( 1 + 4 \exp (-2 \pi^2 M_2^2 h / k^2) )\sum_{|j|>M} \hat{p}(y_j,t_n) \\ 
 + ( 1 + 4 \exp (-2 \pi^2 M_2^2 h / k^2) ) \sum_{|j| \leq M} |r(y_j,t_n)|.
\end{multline}
This is of the form
\begin{equation}
\label{eqn:genform}
R_{n+1} \leq \alpha \pi_n + \alpha R_n.
\end{equation}
We derive from this the sequence of inequalities $\alpha R_n \leq
\alpha^2 \pi_{n-1} + \alpha^2 R_{n-1}$, $\cdots$, $\alpha^{n-1} R_2 \leq
\alpha^{n} \pi_1 + \alpha^{n} R_1$.  Summing these together with (\ref{eqn:genform}), we derive
$$
R_{n+1} \leq \sum_{i=1}^n \alpha^i \pi_{n - i + 1} + \alpha^n R_1.
$$
Applying this to (\ref{eqn:ourform}) and using $r(\cdot,t_1) \equiv
0$, we have
\begin{multline}
\label{eqn:restimate}
\sum_{|i| \leq M} |r(x_i,t_{n+1})| \\
\leq \sum_{i=1}^n (1 + 4 \exp(-2 \pi^2 M_2^2 h/k^2))^i \sum_{|j| > M} \hat{p}(y_j,t_{n-i+1}).
\end{multline}
Now we use (\ref{eqn:qdef}) and the Chernoff bound to derive:
\begin{align*}
\sum_{|j| > M} \hat{p}(y_j,t_{n-i+1}) &= K_{n-i+1} \sum_{|j| > M}
  q(y_j,t_{n-i+1}) \\
 &\leq K_{n-i+1} \left[ P(Q_{n-i+1} \geq y_M) + P(Q_{n-i+1} \leq -y_M)
  \right] \\
 &\leq K_{n-i+1} e^{-s y_M} \left( E[e^{s Q_{n-i+1}}] + E[e^{-s
   Q_{n-i+1}}] \right)
\end{align*} 
We apply Lemma \ref{lem:chernoff} to obtain
\begin{multline}
\label{eqn:coshterm}
k \sum_{|j| > M} \hat{p}(y_j,t_{n-i+1}) \leq \frac{3}{2} K_{n-i+1} e^{-s y_M} \\
 \times \exp \left[ T \left( \frac{M_3^2 s^2}{2}
   + |f(0) s| \right) \right] \left( 1 + 2 \cosh (C s e^{M_1 T} )
\right)
\end{multline}
Applying this result to (\ref{eqn:restimate}), we have
\begin{multline*}
k \sum_{|i| \leq M} |r(x_i,t_{n+1})| \leq  \frac{3}{2} e^{-s y_M} \\
 \times \exp \left[ T
  \left( \frac{M_3^2 s^2}{2} + |f(0) s| \right) \right] \left( 1 + 2 \cosh (C s
  e^{M_1 T} ) \right) \\
\times \sum_{i=1}^n (1 + 4 \exp(-2 \pi^2 M_2^2 h/k^2))^i K_{n-i+1}.
\end{multline*}
By (\ref{eqn:kndef}) and (\ref{eqn:ell1estimate}), we have
$$
K_{n-i+1} \leq \| \hat{p}(\cdot,t_1) \|_{\ell^1} (1 + 4 \exp(-2 \pi^2 M_2^2
h/k^2))^{n-i}
$$
Using this and $n \leq N = T/h$,
\begin{multline}
\label{eqn:lastineq}
k \sum_{|i| \leq M} |r(x_i,t_{n+1})| \leq  \frac{3}{2} e^{-s y_M} \\ 
\times \exp \left[ T
  \left( \frac{M_3^2 s^2}{2} + |f(0) s| \right) \right] \left( 1 + 2 \cosh (C s
  e^{M_1 T} ) \right) \\ \times  \| \hat{p}(\cdot,t_1) \|_{\ell^1} \frac{T}{h} (1 + 4 \exp(-2 \pi^2 M_2^2 h/k^2))^{T/h}.
\end{multline}
Let $s = 1$.  Note that
$$
\lim_{h \to 0} (1 + 4 \exp(-2 \pi^2 M_2^2 h/k^2))^{T/h} = 1
$$
and $\lim_{k \to 0} k \| \hat{p}(\cdot,t_1) \|_{\ell^1} = 1$.
Thanks to (\ref{eqn:ourM}), we know that $y_M \geq (\varepsilon + \rho
+ 1) (-\log h)$.  We have shown the right-hand side of
(\ref{eqn:lastineq}) behaves like
$$h^{\varepsilon + \rho + 1} k^{-1} h^{-1} = h^{\varepsilon} = O(h),$$
as desired.  \end{proof}

So long as $M$ remains a positive integer, we can add a
constant to (\ref{eqn:ourM}) and still prove Lemma \ref{lem:rl1}.
What is important is how $M$ scales as a function of $h$; the
logarithmic rate given in (\ref{eqn:ourM}) is the rate at
which we have to push $M$ to $+\infty$ so that we obtain $O(h)$
convergence.  If we push $M$ to $+\infty$ at a faster rate, e.g., by
replacing $(-\log h)$ with $h^{-1}$, then $r$ will converge at a rate
that is exponential in $h$.

Thus far we have considered convergence of $r$ in a truncated and
scaled version of the $\ell^1$ norm.  Convergence in $L^1$ is an easy consequence.
\begin{lemma}
\label{lem:rL1}
Suppose $f$ and $g$ are admissible in the sense of Definition
\ref{def:admissible}.  Suppose $k = h^\rho$ for $\rho > 1/2$,
and that $h, k > 0$ satisfy (\ref{eqn:newhk1}).
For $\varepsilon \geq 1$, let $M$ be defined as in (\ref{eqn:ourM}).
Let $h_\ast$ be defined as in Lemma \ref{lem:chernoff}.
Then for $h < h_\ast$, we have $\left\| r(\cdot,T) \right\|_1 = O(h)$.
\end{lemma}
\begin{proof}
Note that
$$
|r(x,T)| \leq k \sum_{|j|>M} G(x,y_j) \hat{p}(y_j,t_{N-1}) 
 + k \sum_{|j| \leq M} G(x,y_j) |r(y_j,t_{N-1})|.
$$
This is similar to what we wrote above, except that the discrete variable $x_i$
has been replaced by the continuous variable $x$.  We now integrate both sides
with respect to $x$ to obtain
$$
\| r(\cdot,T) \|_1 \leq k \sum_{|j|>M} \hat{p}(y_j,t_{N-1}) + k \sum_{|j| \leq
  M} |r(y_j,t_{N-1})|.
$$
The second term is $O(h)$ by Lemma \ref{lem:rl1}.  For the
first term, we use (\ref{eqn:coshterm}) to write
\begin{multline}
\label{eqn:firstterm}
k \sum_{|j| > M} \hat{p}(y_j,t_{N-1}) \leq \frac{3}{2} K_{N-1} e^{-y_M} \\ \times \exp \left[ T \left( \frac{M_3^2}{2} + |f(0)| \right) \right] \left( 1 + 2 \cosh (C e^{M_1 T} )
\right).
\end{multline}
Since $\lim_{k \to 0} k K_{N-1} = 1$ and $e^{-y_M} =
O(h^{\varepsilon+\rho+1})$, the right-hand side of
(\ref{eqn:firstterm}) behaves like $h^{\varepsilon+1} =
O(h^2)$.  \end{proof}

It is now immediately clear that, under certain conditions, we have
established $O(h)$ convergence of $\mathring{p}$ to the true
density $p$ in the $L^1$ norm.
\begin{corollary}
\label{cor:everything}
Suppose that all hypotheses of both Corollary \ref{cor:phatp} and
Lemma \ref{lem:rL1} are satisfied.  Then, combining these results, we have
$\| p(\cdot,T) - \mathring{p}(\cdot,T) \|_1 = O(h)$.
\end{corollary}

\section{Numerical Experiments}
\label{sect:numexp}
In this section, we use \texttt{R} and \texttt{C++} implementations of 
DTQ to study its empirical convergence behavior, and also to
compare against a numerical solver for (\ref{eqn:kolmo}), the Fokker-Planck or
Kolmogorov equation.  All codes described in this section, together
with instructions on how to reproduce Figures \ref{fig:convergence}
and \ref{fig:fpcompare} are available online\footnote{
\href{https://github.com/hbhat4000/sdeinference/tree/master/DTQpaper}{\parbox{3in}{https://github.com/hbhat4000/sdeinference/tree/ \\master/DTQpaper}}}.

\subsection{Convergence}
\label{sect:convergence}
First, we conduct an empirical study of DTQ convergence.  We verify that under the conditions given by Theorem \ref{thm:convergence}, we do observe convergence in practice.  We also show numerical evidence that such convergence takes place when one or more of the hypotheses do not hold.  Each SDE we consider is an equation for a scalar unknown $X_t$.

Let us describe the way in which we conduct numerical tests for each
SDE.  We begin with the initial condition $X_0 = 0$ and solve forward
in time until $T = 1$.  That is, we apply DTQ (\ref{eqn:pcircn1}) to compute
$\mathring{p}(x,1)$.  We use the following values of the temporal step
$h$:
\begin{equation}
\label{eqn:hvals}
\{0.5,0.2,0.1,0.05,0.02,0.01,0.005,0.002,0.001\}.
\end{equation}
For $h \geq 0.01$, we find that an implementation of DTQ
written completely in \texttt{R} is able to run in a reasonable amount of time.
For $h=0.005$ and below,
we use an implementation where computationally intensive parts of the
code are written in \texttt{C++}; this code is glued to our \texttt{R} code using the Rcpp and RcppArmadillo packages \citep{Rcpp2011,Rcpp2013,RcppArmadillo2014,Armadillo2016}.

The remaining algorithm parameters are set in the following way:
\begin{subequations}
\label{eqn:algoparams}
\begin{gather}
\label{eqn:kchoice}
k = h^{3/4} \\
\label{eqn:Mdef}
\begin{cases} \text{All except Ex. (\ref{eqn:ex3})} & M = \lceil \pi / k^2 \rceil \\
 \text{Ex. (\ref{eqn:ex3})} &  M = \lceil \pi/(2k) - 2 \rceil.\end{cases} \\
\label{eqn:xjchoice}
x_j = j k, \text{ for } - M \leq j \leq M.
\end{gather}
\end{subequations}
For each value of $h$, we compare $\mathring{p}(x,T)$ computed using 
DTQ against the exact solution $p(x,T)$.  Let $F(y,T) = \int_{x=-\infty}^{x=y} p(x,T) \, dx$ denote the cumulative distribution function associated with
the density $p$.  Each comparison is carried out using the following three norms:
\begin{subequations}
\label{eqn:allnorms}
\begin{align}
\label{eqn:tvnorm}
\| p(\cdot,T) - \mathring{p}(\cdot,T)\|_1 &\approx k \sum_{j=-M}^{j=M} | p(j
  k,T) - \mathring{p}(j k, T) | \\
\label{eqn:infnorm}
\| p(\cdot,T) - \mathring{p}(\cdot,T)\|_\infty &\approx \sup_{|j| \leq M} | p(j
  k,T) - \mathring{p}(j k, T) | \\
\label{eqn:ksnorm}
\| F(\cdot,T) - \mathring{F}(\cdot,T)\|_\infty &\approx \sup_{|j| \leq M} | F(j
  k,T) - \mathring{F}(j k, T) |
\end{align}
\end{subequations}
For our tests, we consider six SDE examples, all for a scalar unknown $X_t$:
\begin{subequations}
\label{eqn:allexamples}
\begin{align}
\label{eqn:ex1}
&\begin{cases} % used to be Example 3!
      dX_t = -X_t dt + dW_t \\
      p(x,t) = 
\displaystyle \frac{\exp(-x^2/(1 - \exp(-2 t)))}{\sqrt{(\pi (1 - \exp(-2 t)))}} 
\end{cases} \\
\label{eqn:ex2}
&\begin{cases} % used to be Example 2!
      dX_t = -\frac{1}{2} \tanh X_t \sech^2 X_t dt + \sech X_t dW_t \\
      p(x,t) = (2 \pi t)^{-1/2} (\cosh x) \exp(-\sinh^2 x/(2
      t)) \end{cases} \hspace{-5em} \\
\label{eqn:ex3}
&\begin{cases} % used to be Example 6!
      dX_t = -( \sin{X_t} \cos^3{X_t}) dt + (\cos^2{X_t}) dW_t \\
      p(x,t) = (2 \pi t)^{-1/2} (\sec^2 x) \exp(-\tan^2 x/(2 t))
\end{cases} \hspace{-5em} \\
\label{eqn:ex4}
&\begin{cases} % used to be Example 8!
      dX_t = \left( \frac{1}{2}X_t + \sqrt{1 + X_t^2} \right) dt + \sqrt{1 + X_t^2} dW_t \\
      p(x,t) = (2 \pi (1 + x^2))^{-1/2} \\
      \qquad \qquad \times \exp(-(\sinh^{-1} x - t)^2/2)
\end{cases} \hspace{-5em} \\
\label{eqn:ex5}
&\begin{cases} % used to be Example 1!
      dX_t = \frac{1}{2} X_t dt + \sqrt{1 + X_t^2} dW_t \\
      p(x,t) = (2 \pi t (1 + x^2))^{-1/2} \\
      \qquad \qquad \times \exp(-(\sinh^{-1} x)^2/(2 t))
\end{cases} \\
\label{eqn:ex6}
&\begin{cases} % used to be Example 4!
      dX_t = \left(-\sqrt{1+X_t^2} \sinh^{-1} X_t + \frac{1}{2} X_t
      \right) dt \\ \qquad \qquad + \sqrt{1 + X_t^2} dW_t \\
      p(x,t) = 
\displaystyle \frac{ \exp(-(\sinh^{-1} x)^2/(1 - \exp(-2 t))) }{\sqrt{(\pi (1 - \exp(-2 t)) (1 + x^2))}}
\end{cases} \hspace{-5em}
\end{align}
\end{subequations}
Note that for each example, we have supplied an exact solution
in the form of a probability density function $p(x,t)$.  For each example,
we compare the DTQ density with $p(x,T=1)$.

\begin{figure*}[tbh]
{\centering \includegraphics[width=2.45in,clip,trim=0 50 0
  0]{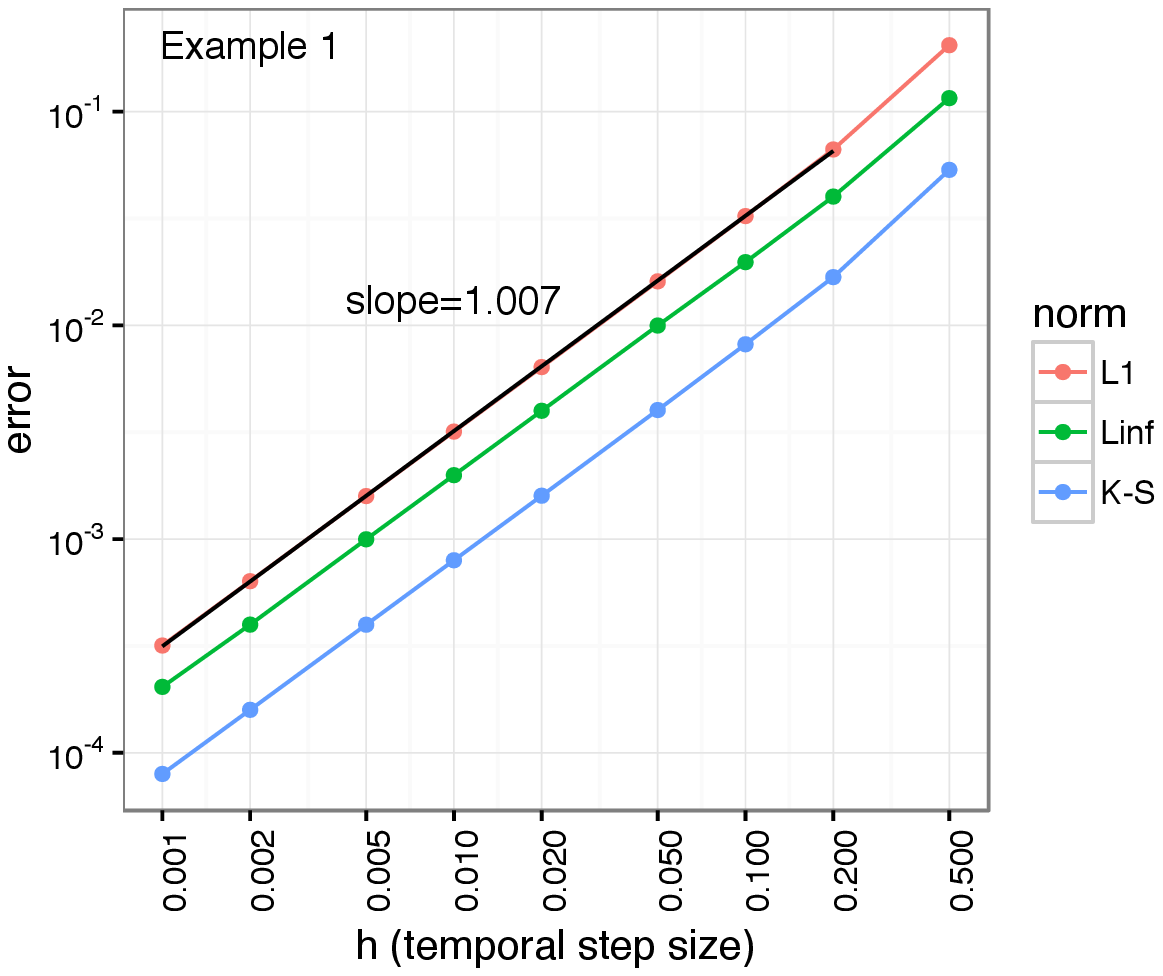} \includegraphics[width=2.45in,clip,trim=0 50 0 0]{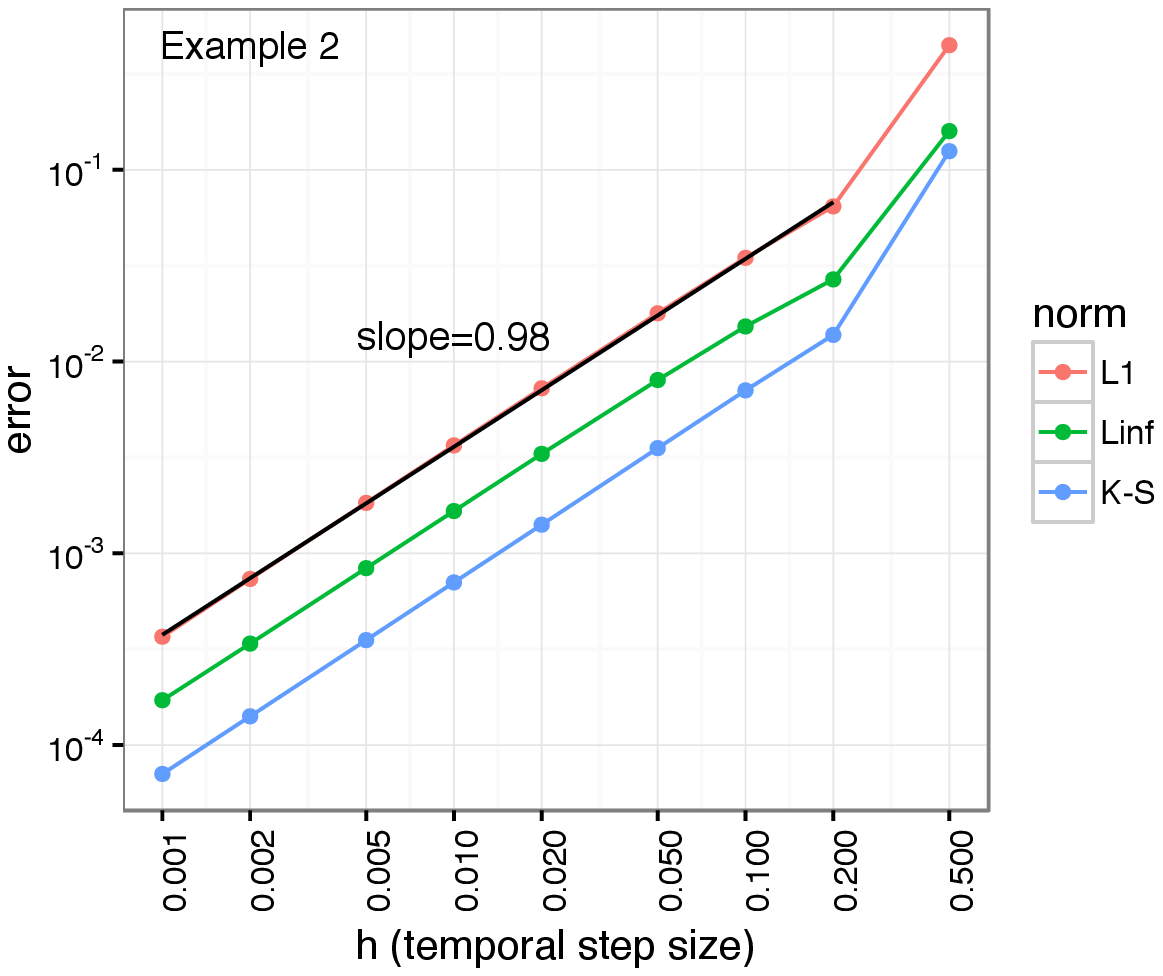}} \\
{\centering \includegraphics[width=2.45in,clip,trim=0 50 0
  0]{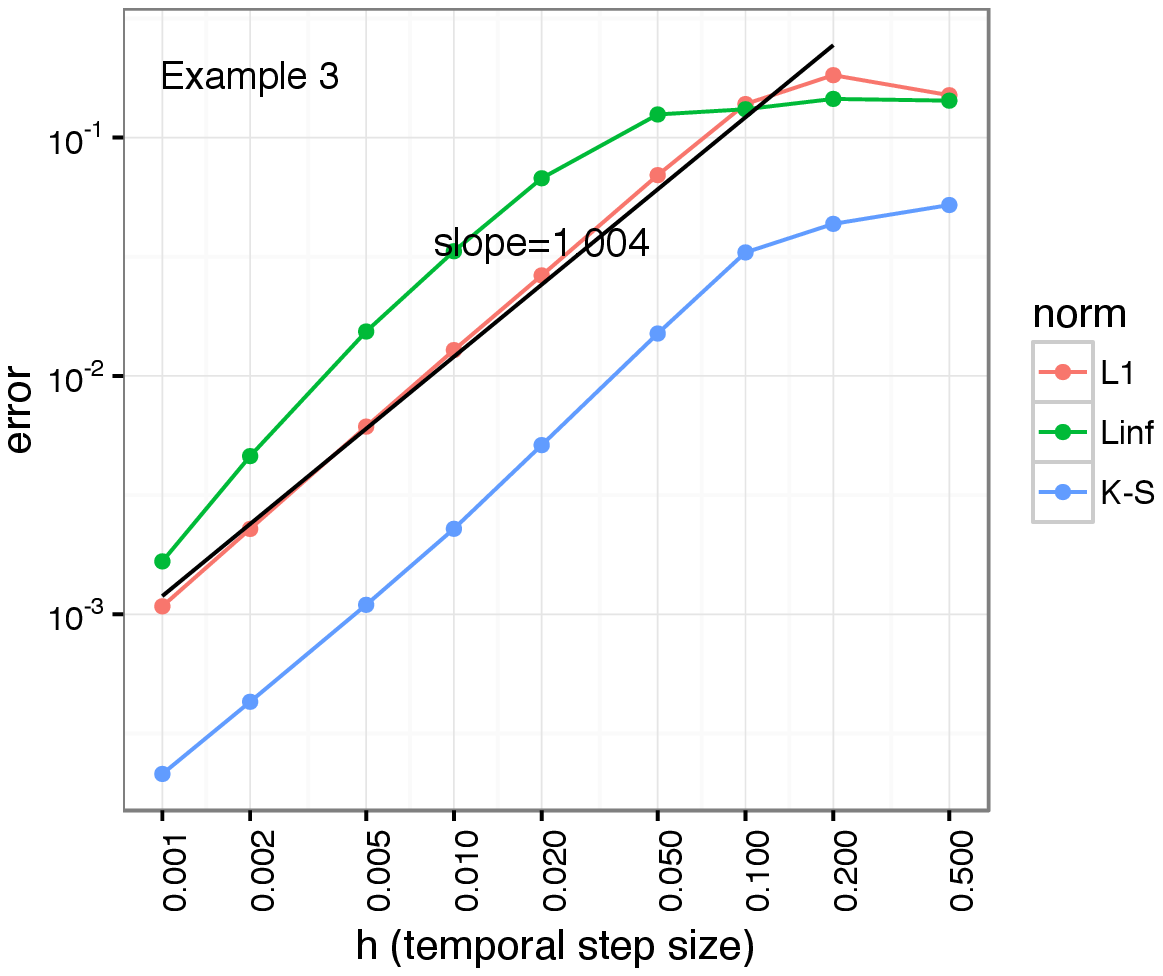} \includegraphics[width=2.45in,clip,trim=0 50 0 0]{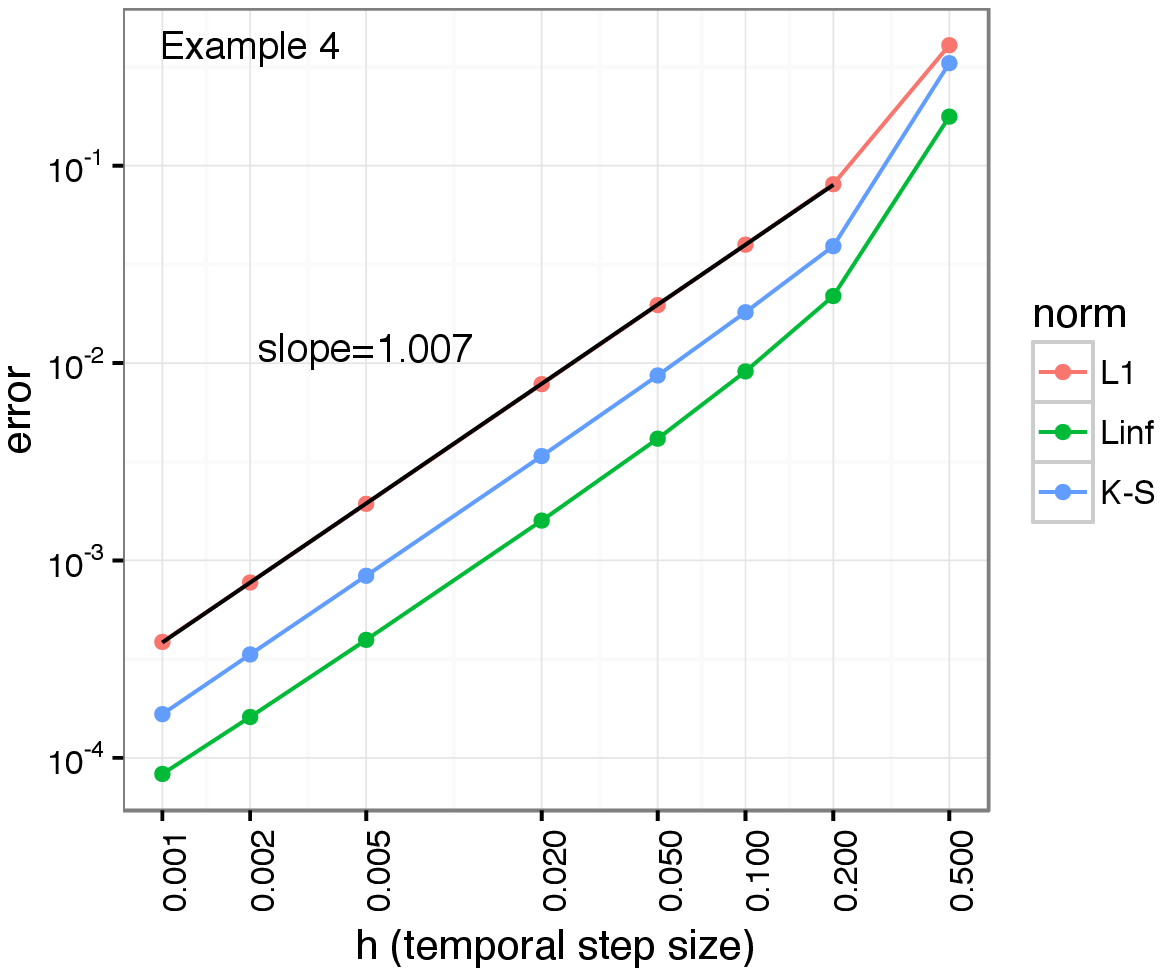}} \\
{\centering \includegraphics[width=2.45in]{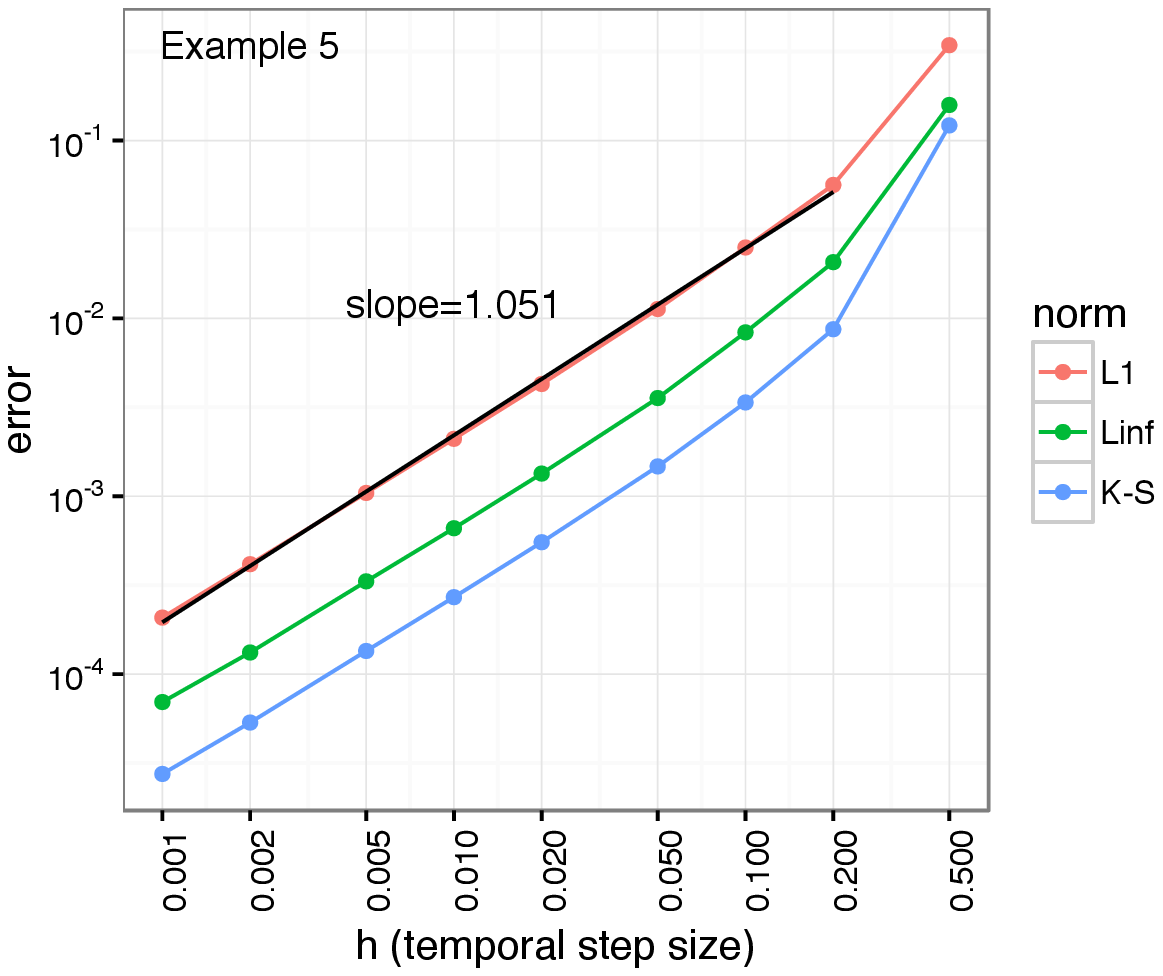} \includegraphics[width=2.45in]{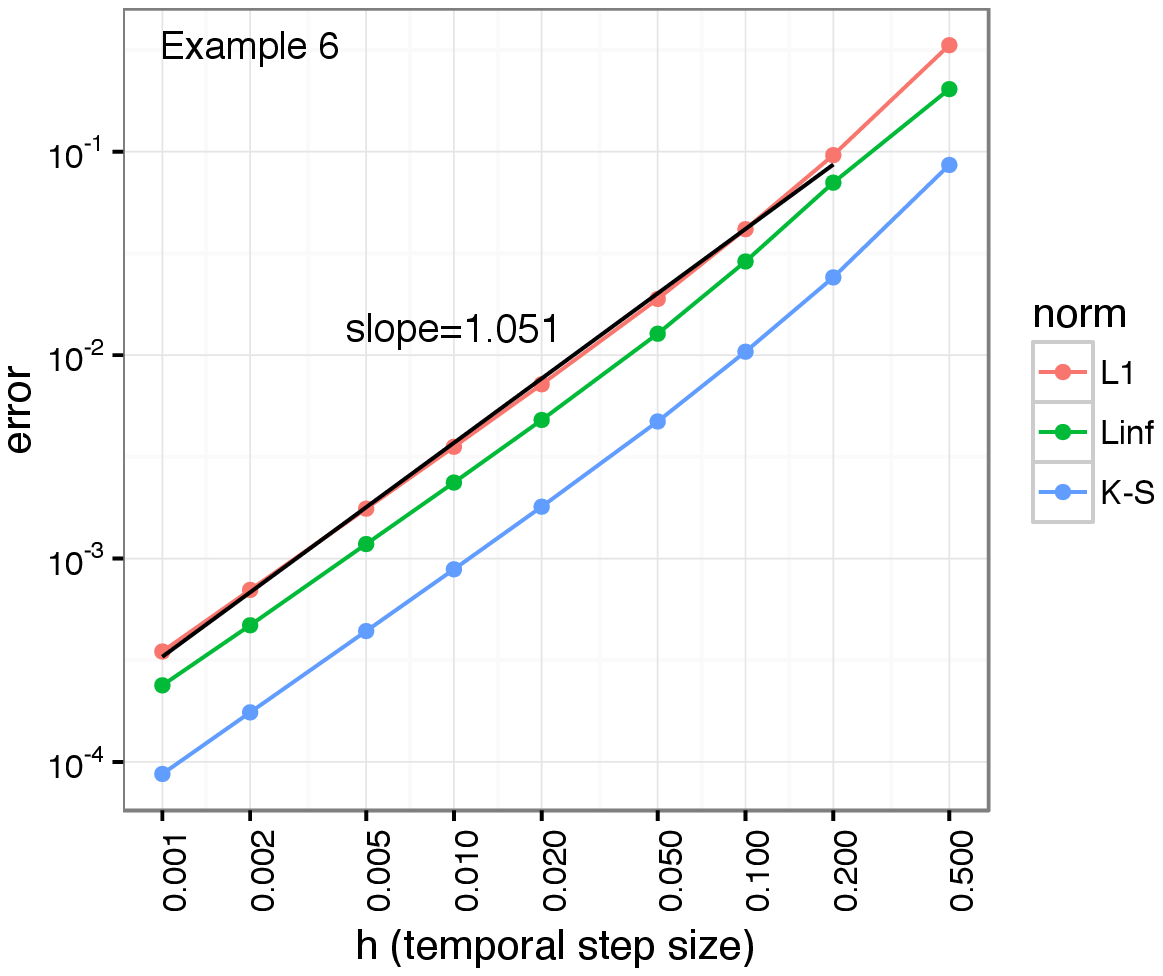}}
\caption{For each of the six examples in
  (\ref{eqn:allexamples}), we test the DTQ method's convergence.  For each
  example, we plot errors between DTQ and exact solutions on log-scaled axes as a function of $h$, the temporal step
  size; all other parameters are given by (\ref{eqn:algoparams}).  We compute errors in each of the three norms given by
  (\ref{eqn:allnorms}).  The horizontal axes (labels and tick mark locations) are the same for
  all plots and correspond to the $h$ values in (\ref{eqn:hvals}).
  Least-squares fits to the $L^1$ error data are indicated by black
  lines and corresponding slope values.  For all examples, we observe first-order convergence, consistent with our $O(h)$ theoretical result.}
\label{fig:convergence}
\end{figure*}

Figure \ref{fig:convergence} shows the convergence results for all six examples.  The overall impression we gain from the plots is that the practical $L^1$ error between the DTQ and exact density functions scales like $h$.  As we now explain, this first-order convergence is displayed under a variety of conditions.

Example (\ref{eqn:ex1}) features drift and diffusion coefficients that clearly satisfy the hypotheses of our convergence theory.  In this case, the computational results confirm the theory.

In Example (\ref{eqn:ex2}), the drift and diffusion coefficients satisfy all but one
of the hypotheses.  Specifically, because $\sech x \to 0$ as $|x| \to
\infty$, the diffusion coefficient is not bounded away from zero.
However, as a matter of numerical practice, on any truncated domain of
the form (\ref{eqn:algoparams}), the diffusion coefficient never
equals zero.  We can say, then, that on the computational domain, the
diffusion coefficient does have a global lower bound that is greater
than zero.  The computational results display first-order convergence.

Example (\ref{eqn:ex3}) is similar to Example (\ref{eqn:ex2}) in that all but one of the
hypotheses are satisfied.  Again, it is the diffusion coefficient
$\cos^2 y$ that is not bounded away from zero.  However, either an
analysis of the original SDE or inspection of the exact solution
reveals that the density will only be supported on the interval
$(-\pi/2, \pi/2)$.  For this SDE, we set $M = \lceil \pi/(2k) - 2
\rceil$ as in (\ref{eqn:Mdef}), retaining (\ref{eqn:kchoice}) and (\ref{eqn:xjchoice}).  This way, the spatial grid covers the interior of $(-\pi/2,\pi/2)$ and the diffusion coefficient never reaches zero.  Again, the computational results show that the $L^1$ error scales like $h$.  % Comparing with the other plots, the black line does not fit the red ($L^1$) errors as well as for other plots.  This is due entirely to using a kernel density estimate (rather than the exact solution) as the reference solution.

Moving to Examples (\ref{eqn:ex4}) and (\ref{eqn:ex5}), the diffusion coefficient is now bounded from below by $1$ but unbounded above.  All other hypotheses of our convergence theory are satisfied.  The empirical convergence rates for both examples match what we expect from theory.

Reexamining the situation with slightly more depth, what we find from our proofs is that (\ref{eqn:usingm3}) is the only place where the upper bound on the diffusion coefficient is used.  However, for the particular case of the diffusion coefficient $g(x) = (1+x^2)^{1/2}$ used in Examples 4 and 5, we have that
\begin{equation}
\label{eqn:stillworks}
\left| \Im \left( g(y + \irm \epsilon) g'(y + \irm \epsilon) \right) \right| = \left| \Im (y + \irm \epsilon) \right| \leq d,
\end{equation}
meaning that we can substitute $d$ for $M_3 M_4$ and the convergence proof follows.  This is an example of how, for specific SDE that do not satisfy the hypotheses of the general theorem, we may yet be able to prove convergence of the DTQ method.

Finally, we come to Example (\ref{eqn:ex6}).  Now we have that the derivative of the drift coefficient is unbounded \emph{and} that the diffusion coefficient is unbounded above.  Though the hypotheses of the convergence theory are not satisfied, we still observe first-order convergence.

For the SDE in Example (\ref{eqn:ex6}), even if we are able to patch our proof to
prove that $\hat{p}$ converges to $\tilde{p}$, we can no longer apply
the result of \citet{BallyTalay1996} to guarantee
convergence of $\tilde{p}$ to $p$.  Overall, we take the numerical
results for Example (\ref{eqn:ex6}) as evidence that $\tilde{p}$ must converge to $p$
under more general conditions than have been established in the literature.

\subsection{Comparison with Fokker-Planck}
\label{sect:fpcomp}
Now we compare DTQ against a classical
approach, that of numerically solving the Fokker-Planck or Kolmogorov
PDE (\ref{eqn:kolmo}).  In what follows, we use subscripts to denote
partial derivatives, so that (\ref{eqn:kolmo}) is written
\begin{equation}
\label{eqn:kolmo2}
p_t + \left(f(x) p(x,t)\right)_x = \frac{1}{2} \left(g^2(x) p(x,t)\right)_{xx}.
\end{equation}

To solve this equation, we employ a standard finite
difference method.  To resolve the singular initial condition $p(x,0)
= \delta(x)$, we use a subtraction technique: we set $p = u + v$,
where $u$ solves
\begin{equation}
\label{eqn:upde}
u_t = \frac{1}{2} \kappa u_{xx}, \qquad u(x,0) = \delta(x),
\end{equation}
while $v$ solves
\begin{subequations}
\label{eqn:vpde}
\begin{align}
\begin{split}
\label{eqn:vpdea}
& v_t + \left(f(x) v(x,t)\right)_x = \frac{1}{2} \left(g^2(x)
  v(x,t)\right)_{xx} \\
 & \qquad + \underbrace{\frac{1}{2} \left[ \left( g^2(x) - \kappa
  \right) u(x,t) \right]_{xx} - \left[ f(x) u(x,t) \right]_{x}}_{F(x,t)} \end{split} \\
&v(x,0) = 0.
\end{align}
\end{subequations}
The point is that (\ref{eqn:upde}) can be solved analytically, i.e.,
for $t > 0$,
\begin{equation}
\label{eqn:usol}
u(x,t) = \frac{1}{\sqrt{2 \pi \kappa t}} \exp \left( -\frac{x^2}{2
    \kappa t} \right).
\end{equation}
Here $\kappa > 0$ is a parameter that we are free to set.  In our own
tests, we use $\kappa = 1$.  Since (\ref{eqn:usol}) is known, we
substitute it into the final two terms on the right-hand side of
(\ref{eqn:vpdea})---this yields a known forcing term $F(x,t)$.  We
then employ the following numerical scheme to solve (\ref{eqn:vpde})
for $v(x,t)$:
\begin{itemize}
\item We discretize $v(x,t)$ on fixed spatial and temporal grids with
  respective spacings $k$ and $h$.  Let
  $V_j^n$ denote our numerical approximation to $v(jk,nh)$.  Here $0
  \leq n \leq N$ with $N h = T > 0$, the final time.  We also have
  that $-M \leq j \leq M$.  Implicitly, we assume that $v(x,t) = 0$
  for $|x| > Mk$.
\item We use a first-order approximation to $v_t$: $v_t(x,t) \approx
  (V_j^{n+1} - V_j^n)/h$.
\item We treat the drift term explicitly: 
$$
\left( f(x) v(x,t) \right)_x \approx \left( f((j+1) k) V_{j+1}^n -
  f((j-1) k) V_{j-1}^n \right)/(2 k).
$$
\item We treat the diffusion term implicitly:
\begin{multline*}
\frac{1}{2} \left( g^2(x) v(x,t) \right)_{xx} \approx \frac{1}{2 k^2}
\Bigl[ g^2((j-1) k) V_{j-1}^{n+1} \\ - 2 g^2(j k) V_{j}^{n+1} + g^2((j+1)
  k) V_{j+1}^{n+1} \Bigr].
\end{multline*}
\end{itemize}
Let $\mathbf{V}^n$ be a vector of length $2M+1$ whose $j$-th entry
is $V_j^n$.  Then, combining approximations, we obtain the
matrix-vector system
\begin{equation}
\label{eqn:fpmatvec0}
A \mathbf{V}^{n+1} = B \mathbf{V}^n + \mathbf{F}^n
\end{equation}
with tridiagonal matrices $A$ and $B$ given by (\ref{eqn:Amat}) and 
(\ref{eqn:Bmat}) in Table \ref{tab:ABmat}.
\begin{table*}[t]
\begin{gather}
\label{eqn:Amat}
A = \begin{bmatrix} 1 + \frac{h}{k^2} g^2_{-M} & -\frac{h}{2k^2}
  g^2_{-M+1} & & & & \\
 -\frac{h}{2k^2} g^2_{-M} & 1 + \frac{h}{k^2} g^2_{-M+1} & -\frac{h}{2k^2}
 g_{-M+2} & & & \\
 & -\frac{h}{2k^2} g^2_{-M+1} & 1 + \frac{h}{k^2} g^2_{-M+2} &
 -\frac{h}{2k^2} g^2_{-M+3} & & \\
& & \diagdots{7em}{0em}
& \diagdots{7em}{0em}
& \diagdots{7em}{0em} & \\
& & & -\frac{h}{2k^2} g^2_{M-2} & 1 + \frac{h}{k^2} g^2_{M-1} &
-\frac{h}{2k^2} g^2_M \\
& & & & -\frac{h}{2k^2} g^2_{M-1} & 1 + \frac{h}{k^2} g^2_M \end{bmatrix} \\
\label{eqn:Bmat}
B = \begin{bmatrix} 1 & -\frac{h}{2k} f_{-M+1} & & & & \\
 \frac{h}{2k} f_{-M} & 1  & -\frac{h}{2k} f_{-M+2} & & & \\
 & \frac{h}{2k} f_{-M+1} & 1  & -\frac{h}{2k} f_{-M+3} & & \\
& & \diagdotss{7em}{0em} & \diagdotss{9em}{0em} & \diagdotss{5em}{0em} & \\
& & & \frac{h}{2k} f_{M-2} & 1  & -\frac{h}{2k} f_M \\
& & & & \frac{h}{2k} f_{M-1} & 1 \end{bmatrix}.
\end{gather}
\caption{Tridiagonal matrices used in the Fokker-Planck solver---see 
(\ref{eqn:fpmatvec0}) for further details.}
\label{tab:ABmat}
\end{table*}
We also define $\mathbf{F}^n$ in (\ref{eqn:fpmatvec0}) by discretizing
$F(x,t)$ in (\ref{eqn:vpdea}).  That is, for $-M \leq j \leq M$,
we define the $j$-th component of $\mathbf{F}^n$ by
\begin{multline}
\label{eqn:fndef}
F_j^n = \frac{h}{2k^2} \Bigl[ g^2((j-1)k) u((j-1)k,nh) \\
 - 2 g^2(j k) u(jk,nh) + g^2((j+1)k) u((j+1)k,nh) \Bigr] \\ 
 - \frac{h}{2k} \Bigl[ f((j+1)k) u((j+1)k,nh) - f((j-1)k) u((j-1)k,nh) \Bigr].
\end{multline}
To solve for $\mathbf{V}^{n+1}$ given $\mathbf{V}^{n}$, we rewrite
(\ref{eqn:fpmatvec0}) as
\begin{equation}
\label{eqn:fpmatvec}
\mathbf{V}^{n+1} = A^{-1} B \mathbf{V}^n + A^{-1} \mathbf{F}^n.
\end{equation}
We compute
$u(jk,T)$ for $-M \leq j \leq M$ and denote the resulting vector by
$\mathbf{U}^N$.
Let $p_\text{FP}(\mathbf{x},T)$ denote the
vector whose $j$-th component is $p_\text{FP}(x_j,T)$, the
approximation of $p(x_j,T)$ obtained by solving the Fokker-Planck
equation numerically.  With these definitions, our algorithm for
computing $p_\text{FP}$ is easily stated: we start with $\mathbf{V}^0
= \mathbf{0}$, iterate (\ref{eqn:fpmatvec}) $N$ times to compute
$\mathbf{V}^N$, and then compute
$$
p_\text{FP}(\mathbf{x},T) = \mathbf{U}^N + \mathbf{V}^N.
$$
Note that in our Fokker-Planck solver, the
matrices $A$ and $B$ defined by (\ref{eqn:Amat}) and (\ref{eqn:Bmat})
are implemented as sparse tridiagonal matrices.  When we use
(\ref{eqn:fpmatvec}) to solve for $\mathbf{V}^{n+1}$, we use sparse
numerical linear algebra to compute both $A^{-1} B$ and $A^{-1}
\mathbf{F}^n$.  In particular, $A^{-1} B$ is precomputed before we
loop from $n=0$ to $n=N-1$.
\begin{figure*}[tbh]
{\centering \includegraphics[width=4.95in]{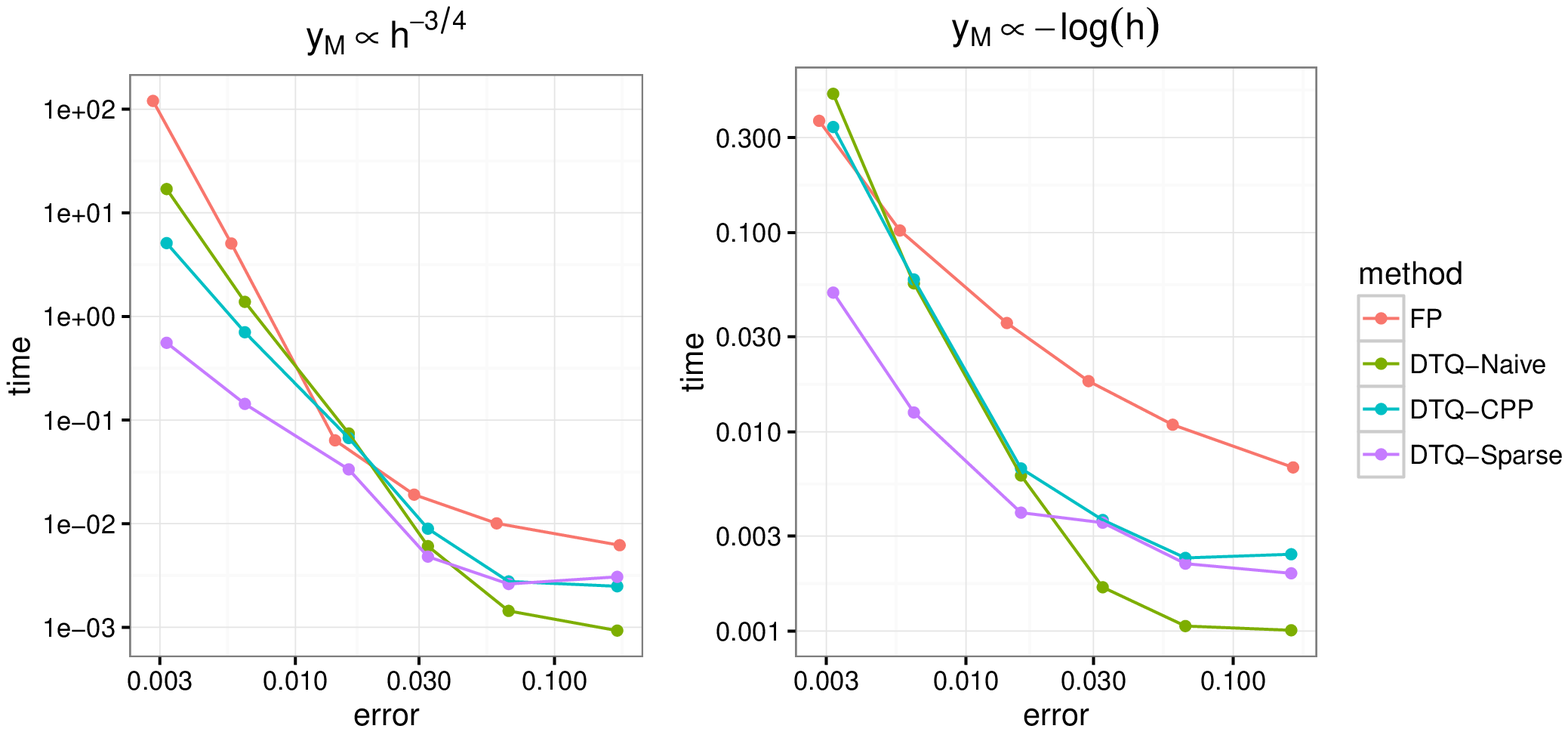}}
\caption{For a particular SDE, Example (\ref{eqn:ex1}), suppose we are interested in computing
  the density $p(x,T)$ at time $T=1$.  When we compute this density,
  we will incur some error, measured here in the $L^1$ norm.  \textbf{The
  plotted results show that for a fixed value of this error, the
  DTQ methods require less computational time (measured in wall clock
  seconds) than a method for numerically solving the Fokker-Planck
  PDE.} In all simulations, we use a domain $[-y_M, y_M]$.  For the
simulations in the left (respectively, right) plot, we have scaled the
domain according to $y_M \propto h^{-3/4}$ (respectively, $y_M \propto
\log h^{-1}$), where $h > 0$ is the time step.  In both plots, we see that for smaller values of the
error, the fastest method is DTQ-Sparse; for larger values of
the error, the fastest method is DTQ-Na\"ive.  In particular, for the
smallest error of $0.003$, DTQ-Sparse is over $10^2$
(respectively, $10^{3/4}$) times
faster than the Fokker-Planck method in the left (respectively, right)
plot.  Despite the fact that our
Fokker-Planck solver uses the same sparse numerical linear algebra as
DTQ-Sparse, it is often the slowest of the four methods. For details regarding
the three implementations of DTQ (DTQ-Na\"ive, DTQ-CPP,
and DTQ-Sparse) as well as the implementation of our Fokker-Planck
solver, please see Section \ref{sect:fpcomp}. }
\label{fig:fpcompare}
\end{figure*}

We are now in a position to compare the DTQ and Fokker-Planck
methods.  For this comparison, we exclusively use the drift and
diffusion functions from Example (\ref{eqn:ex1}).  As
described above, among the examples in (\ref{eqn:allexamples}),
Example (\ref{eqn:ex1}) is the only one that satisfies all of the hypotheses of our
DTQ convergence theory.

As mentioned in Section \ref{sect:trunc},
  when we implement DTQ in practice, we start with
  (\ref{eqn:pcircn1})---with $x$ discretized on the same spatial grid
  as $y$, i.e.,
\begin{equation}
\label{eqn:actualdtqmethod}
\mathring{p}(x_i,t_{n+1}) = k \sum_{j=-M}^M G(x_i,y_j) \mathring{p}(y_j,t_n)
\end{equation}
For fixed $n$, as $j$ varies from $-M$ to $M$, the elements
$\mathring{p}(y_j,t_n)$ form a $(2M+1)$-dimensional vector that we denote
$\mathbf{p}^n$.  With this notation, (\ref{eqn:actualdtqmethod}) can
be written
\begin{equation}
\label{eqn:actualdtqmethod2}
\mathbf{p}^{n+1} = \mathcal{A} \mathbf{p}^n,
\end{equation}
where $\mathcal{A}$ is the $(2M+1) \times (2M+1)$ matrix whose $(i,j)$-th
element is $k G(x_i,y_j)$.  In our experience, \emph{the most
 computationally expensive part of DTQ is the assembly of
 $\mathcal{A}$.}  For the tests presented in this subsection, we have
implemented three different methods to compute $\mathcal{A}$:
\begin{enumerate}
\item \textbf{DTQ-Na\"ive}.  Here we assemble $\mathcal{A}$ using dense matrix
  methods in \texttt{R}.  The main advantage of this approach is ease of
  implementation; the code to compute $\mathcal{A}$ is only $4$ lines long.
  Incidentally, the convergence tests in the first part of this
  section use DTQ-Na\"ive for $h \geq 0.01$.
\item \textbf{DTQ-CPP}.  Implicitly, DTQ-Na\"ive forces \texttt{R} 
  to loop over the entries of $\mathcal{A}$ serially. In DTQ-CPP, we
  use Rcpp together with OpenMP directives in \texttt{C++} to compute and fill in the entries of
  $\mathcal{A}$ in parallel.   In practice, we run this code on a machine with
  $12$ cores, setting the number of OpenMP threads to $12$.
\item \textbf{DTQ-Sparse}.  Here we take advantage of the structure of
  $\mathcal{A}$.  Specifically, we have
$$
\mathcal{A}_{ij} = k G(x_i,y_j) = \frac{k}{\sqrt{2 \pi g^2(y_j) h}} \exp \left(
  -\frac{(x_i - y_j - f(y_j) h)^2}{2 g^2(y_j) h} \right).
$$
Let us set $i = j + i'$.  Then we have
\begin{equation}
\label{eqn:Adiag}
\mathcal{A}_{j+i', j} = \frac{k}{\sqrt{2 \pi g^2(y_j) h}} \exp \left(
  -\frac{(i'k - f(y_j) h)^2}{2 g^2(y_j) h} \right).
\end{equation}
We think of $i'$ as indexing the sub-/super-diagonals of $\mathcal{A}$.
For each fixed $i' = 0, 1, 2, \ldots$ we evaluate
(\ref{eqn:Adiag}) over all $j$ to obtain the $i'$-th
subdiagonal of $\mathcal{A}$.  For $h$ small, as $i'$ increases, we observe
that the entire subdiagonal decays rapidly.  In our implementation, we
compute subdiagonals until the $1$-norm of the subdiagonal drops below
$2.2 \times 10^{-16}$ (machine precision in \texttt{R}) multiplied by the
$1$-norm of the main $i'=0$ diagonal of $\mathcal{A}$.  We then compute the same
number of superdiagonals as subdiagonals.  The final $\mathcal{A}$ matrix is
assembled as a sparse matrix using the CRAN Matrix package
\citep{MatrixCRAN}.
\end{enumerate}
Given the tridiagonal structure of both $A$ and $B$ in the Fokker-Planck
method, we do not believe any reasonable modern implementation would use
dense matrices.  Similarly, while DTQ-Na\"ive requires minimal
programming effort, a reasonable implementation would look much more
like DTQ-CPP or DTQ-Sparse.  None of the DTQ methods require more
programming effort to implement than the Fokker-Planck method.

\paragraph{Results for $O(h^{3/4})$ Domain Scaling.} For each $h$ in
(\ref{eqn:hvals}) that satisfies $h \geq 0.01$, we use all three DTQ methods and the Fokker-Planck method to generate
numerical approximations of the density function at the final time
$T=1$.  For our first set of comparisons, parameters such as $k$ and
$M$ are set via (\ref{eqn:algoparams}).  In particular, the
computational domain is $[-y_M, y_M]$ where $y_M = M k \propto h^{-3/4}$.
We compute the $L^1$ errors between each numerical solution and the exact solution $p(x,T)$.  We also record
the wall clock time (in seconds) required to compute the solution
using each method.  Each measurement is repeated $100$ times; we
report average results.

In the left panel of Figure \ref{fig:fpcompare}, we have plotted (on
log-scaled axes) wall clock time as a function of $L^1$ error for
each of the four methods.  We see that if one can tolerate a
relatively large $L^1$ error, then the fastest method is the
DTQ-Na\"ive method (green); for $L^1$ errors less than $0.03$, the
fastest method is DTQ-Sparse (purple).  The Fokker-Planck
method is often the slowest of the four methods.  For an error of
$0.003$, DTQ-Sparse is approximately $100$ times faster
than the Fokker-Planck method.

\paragraph{Results for $O(\log h^{-1})$ Domain Scaling.}  For our second set of comparisons, we have changed the way that $y_M$
(effectively, the size of the computational domain) scales with $h$.
We retain $k = h^{3/4}$ but now set $y_M = (2 + 3/4)(-\log h) \propto
(-\log h)$ in accordance with (\ref{eqn:ourM}).  The
spatial grid, for all four methods, is now given by $x_j = -y_M +
(j+M)k$ for $-M \leq j \leq M$ with $M = \lfloor y_M / k \rfloor$.  In
all other respects, we make no changes and rerun the test described
above for all four methods.

In the right panel of Figure \ref{fig:fpcompare}, we
have plotted (on log-scaled axes) wall clock time as a function of
$L^1$ error for each of the four methods.  Once again, we find that
DTQ-Na\"ive and DTQ-Sparse are the fastest for,
respectively, large and small error values.  For an error of $0.003$,
DTQ-Sparse is approximately $10^{3/4} \approx 5.62$ times
faster than the Fokker-Planck method.

\section{Conclusion and Future Directions}
\label{sect:conclusion}
We have established fundamental properties of the DTQ method, including theoretical and empirical convergence results.  Let us make three concluding remarks regarding our results.

First, we have not yet mentioned that DTQ features two properties that are not always easy to establish for numerical methods for the Fokker-Planck equation (\ref{eqn:kolmo}): (i) DTQ automatically preserves the
nonnegativity of the computed density $\hat{p}$, and (ii) the DTQ
density $\hat{p}$ has a normalization constant that can be estimated
for finite $h, k > 0$.  In practice, we find that $\mathring{p}$ is very
close to being correctly normalized.

Second,  $p(x,T)$ and $\tilde{p}(x,t_N)$ correspond to,
  respectively, the random variables $X_T$ and $x_N$.  Convergence in
  $L^1$ of $\tilde{p}$ to $p$ is equivalent to convergence in total
  variation of $x_N$ to $X_T$.  Note that
\begin{equation}
\label{eqn:Gintconseq}
\int_{x \in \mathbb{R}} \hat{p}(x,t_{n+1}) \, dx = k
\sum_{j=-\infty}^\infty \hat{p}(y_j,t_n) = k K_n,
\end{equation}
implying that $\hat{q}(x,t_{n+1}) = \hat{p}(x,t_{n+1})/ (k K_n)$ is the
density of a continuous random variable $y_n$.  An easy consequence
of our results is that $\hat{q}$ converges to $\tilde{p}$ in $L^1$,
implying convergence of $y_N$ to $x_N$ in total variation.

Third, if we trace back the crux of our convergence proof, a key step
is estimating the $L^1$ error of $\tau$ starting from the trapezoidal
rule error estimate (\ref{eqn:tauupper}).  To do this, it was
essential that we have an estimate of $\mathcal{N}$ that is an $L^1$
function of $x$.  It was to obtain such an estimate that we put our
efforts into Lemma \ref{lem:Gbound}.  We have tried to replicate this
analysis using more conventional error estimates for the trapezoidal
rule---estimates that require less regularity of the
integrand than we have assumed.  Thus far, these other attempts have
failed because they do not yield an upper bound on $\tau$ that is
itself an $L^1$ function of $x$.  The approach in the present work is
the only one that we have gotten to work.

The present research motivates four main questions that we seek to answer
in future work:
\begin{enumerate}
\item When we derived the DTQ method, we used three approximations:
  (i) an Euler-Maruyama approximation of the original SDE, (ii) a
  trapezoidal quadrature rule, and (iii) a finite dimensionalization
  of $\tilde{p}$ that consists of sampling the function on a
  truncated grid.  The first question to ask is: what happens to the
  DTQ method if we improve upon these initial approximations?

  Regarding (ii), we can say that we have written a test code in which
  we use Gauss-Hermite quadrature instead of the trapezoidal rule.
  This does not yield better convergence.  Given the exponential
  convergence of $\hat{p}$ to $\tilde{p}$ established here,
  this should not be a surprise.

  Regarding (iii), rather than sampling the function
  $\tilde{p}(x,t_n)$ on a discrete grid, we could have
  instead chosen to represent $\tilde{p}(x,t_n)$ as a linear
  combination of functions---for instance, a linear combination of
  Gaussian densities, where each density is centered at a grid point
  $x_j$.  In a collocation scheme, we would then insert these
  approximations of $\tilde{p}$ into (\ref{cdt}) and enforce equality
  at a finite number of points.  We have tried this as well in a test
  code.  While such a scheme does not yield
  better numerical behavior, it may be easier to
  analyze.

  If we had to choose one approximation (among (i), (ii), or (iii)) to target,
  we would choose (i).
  Suppose we replace the Euler-Maruyama 
  method with a higher-order method.  The higher-order method then 
  induces a new conditional density function $\tilde{p}_{n+1|n}$ that replaces the 
  Gaussian kernel $G$.
  Using this new $\tilde{p}_{n+1|n}$ in place of $G$, the evolution equation 
  (\ref{eqn:pcircn1}) for $\mathring{p}$ remains the same.   Preliminary results
  with the weak trapezoidal method \citep{AndersonMattingly} indicate that, 
  in this way, we can obtain a version of the DTQ method that features 
  $O(h^2$) convergence of $\mathring{p}$ to $p$.  
Note that if we instead retain approximation (i) and replace (ii) and/or (iii), we will be stuck with the $O(h)$ convergence rate of $\tilde{p}$ to $p$, 
thereby blocking improvements to the overall convergence rate of $\mathring{p}$ to $p$.

\item Can we patch DTQ to handle diffusion functions $g$
  that equal zero at, say, a finite number of discrete points in the
  computational domain?  We believe there should be some way of doing
  this by subtracting out singularities of $G$ inside the
  Chapman-Kolmogorov equation (\ref{eqn:phatn1}).

\item Can we derive DTQ-like methods for stochastic differential
  equations driven by stochastic processes other than the Wiener
  process?  In ongoing work,
  we are studying how to derive such methods to solve for the
  density in the case when we replace $dW_t$ by a process whose
  increments follow a L\'{e}vy $\alpha$-stable distribution.  For such
  an SDE, current methods for computing the density involve numerical
  solution of a fractional Fokker-Planck equation.  We expect DTQ-like
  methods to be highly competitive for such problems.
\end{enumerate}

%\item How can we further apply DTQ to problems of statistical
%  inference?  In a typical inference problem, we seek to use data to infer
%  parameters in the drift and diffusion functions.  In preliminary
%  work \citep{BhatMadu2016, BhatMaduRawat2016}, 
%  we have shown how DTQ can
%  be used to efficiently compute two quantities that are important for
%  inference: the likelihood function and its
%  gradient with respect to the parameters.  Further
%  improvements to and generalizations of DTQ, as described
%  above, will yield improved inference algorithms.

\section*{Acknowledgements}
H.S.B. acknowledges computational time on the MERCED cluster (NSF ACI-1429783).  H.S.B. and R.W.M.A.M. acknowledge support for this work from UC Merced, through UC Merced Committee on Research grants, Applied Mathematics Graduate Group fellowships, and a School of Natural Sciences Dean's Distinguished Scholars Fellowship. 

\bibliographystyle{acmtrans-ims}
\bibliography{dtqpaper}

\end{document}